\documentclass[12pt]{article}
\usepackage[top=1in, bottom=1.25in, left=1.25in, right=1.25in]{geometry} 
\usepackage{color}
\usepackage{textcomp}
\usepackage{verbatim}
\usepackage{amsthm}
\usepackage{amsmath}
\usepackage{amssymb}
\usepackage{graphicx}
\usepackage{esint}
\usepackage{authblk}

\makeatletter

\usepackage{fancyhdr}
\usepackage{sectsty}\allsectionsfont{\sffamily\mdseries\upshape} 
\usepackage{amsthm}\usepackage{caption}\usepackage{subcaption}\usepackage{comment}\usepackage{appendix}
\usepackage{todonotes}
\usepackage{mathabx}
\theoremstyle{plain} \newtheorem{theorem}{Theorem}
\theoremstyle{definition} 
\theoremstyle{plain} \newtheorem{lemma}{Lemma}

\theoremstyle{plain} \newtheorem{prop}{Proposition}

\date{}

\makeatother

\begin{document}

\title{Trapped modes in armchair graphene nanoribbons}
\author[1]{V. A. Kozlov\thanks{vladimir.kozlov@liu.se}}
\author[2]{S. A. Nazarov\thanks{srgnazarov@yahoo.co.uk}} 
\author[1]{A. Orlof\thanks{a.orlof@gmail.com}}

\affil[1]{\textit {Mathematics and Applied Mathematics, MAI, Link\"{o}ping University, SE-58183 Link\"{o}ping, Sweden}}
\affil[2]{\textit{ Saint-Petersburg State University, Universitetsky pr., 28, Peterhof, 198504, St. Petersburg, Russia,
Peter the Great St. Petersburg Polytechnic University, St. Petersburg 195251, Russia,
 Institute of Problems of Mechanical Engineering
RAS, V.O., Bolshoi pr., 61, 199178, St.-Petersburg, Russia}}
\maketitle
\begin{abstract}
We study scattering on an ultra-low potential in armchair graphene
nanoribbon. Using the continuous Dirac model and including a couple
of artificial waves in the scattering process, described by an augumented
scattering matrix, we derive a condition for the existence of a trapped
mode. We consider the threshold energies, where the the multiplicity
of the continuous spectrum changes and show that a trapped mode may
appear for energies slightly less than a threshold and its multiplicity
does not exceed one. For energies which are higher than a threshold,
there are no trapped modes, provided that the potential is sufficiently
small. 

\end{abstract}

\section{Introduction}

The very high quality of graphene samples \cite{Chen,Novoselsov} allows us to consider
the production of defects deliberately. There are two types of defects:
short- and long-range. Vacancies and adatoms are classified as a short-range
type and are modelled by Dirac-delta functions. On the other hand,
electric or magnetic fields, interactions with the substrate, Coulomb charges,
ripples and wrinkles can be considered as long-range disorder and
modelled by smooth functions (a Gaussian for example). In the present
study we assume that graphene is free of short-range defects and we
consider only long-range defects described by an external potential.

We work within the continuous Dirac model, where electrons dynamics can
be described by a system of 4 equations \cite{CastroNeto}
\begin{equation}
\mathcal{D}\left(\begin{array}{c}
u\\
v\\
u'\\
v'
\end{array}\right)+\delta\mathcal{P}\left(\begin{array}{c}
u\\
v\\
u'\\
v'
\end{array}\right)=\omega\left(\begin{array}{c}
u\\
v\\
u'\\
v'
\end{array}\right),\label{eq:DiractotINIT}
\end{equation}
with
\begin{equation}
\mathcal{D}=\mathcal{D}(\partial_{x},\partial_{y}):=\left(\begin{array}{cccc}
0 & i\partial_{x}+\partial_{y} & 0 & 0\\
i\partial_{x}-\partial_{y} & 0 & 0 & 0\\
0 & 0 & 0 & -i\partial_{x}+\partial_{y}\\
0 & 0 & -i\partial_{x}-\partial_{y} & 0
\end{array}\right),\label{def:dirac}
\end{equation}
where $(x,y)$ are dimentionless cartesian coordinates (obtained by
the change of coordinates $(x,y)\mapsto\frac{4}{3\sqrt{3}a_{CC}}(x,y)$
in the original problem \cite{CastroNeto}, where $a_{CC}=0.142${[}nm{]}
is a distance between nearest carbon atoms) and consequently $\omega$
is a dimensionless energy (where the original energy in {[}eV{]} can
be recovered by the multiplication by $\frac{2t}{\sqrt{3}}$,where $t=2.77[eV]$
is nearest-neighbour hopping integral); the potential $\mathcal{P}$
is a dimensionless, real-valued function with a compact support and
$\delta$ is a real-value small parameter.

In the discrete model graphene hexagonal lattice is described as a composition
of two interpenetrating triangular lattices (called A and B). The
consequence of this division is a two component wave function $(\psi_{A},\psi_{B})$,
where $\psi_{A}$ ($\psi_{B}$) describes the electron on the sites of
lattice A (B). In the discrete model there are two minima in the dispersion
relation, called $\mathbf{K}=(0,-K)$ and $\mathbf{K}'=(0,K)$ valleys
(with $K=\pi$ in our dimentionless formulation). Low energy approximation
($a_{CC}\rightarrow0$), which enables the passage from discrete to
continuous model, has to be done close to those minima separately, leading
to the following form of the total wave functions \cite{Macucci}:
\begin{equation}
\psi_{A}(x,y)=e^{i\mathbf{K}\cdot(x,y)}u(x,y)-ie^{i\mathbf{K}'\cdot(x,y)}u'(x,y),\label{eq:wfA}
\end{equation}
\begin{equation}
\psi_{B}(x,y)=-ie^{i\mathbf{K}\cdot(x,y)}v(x,y)+e^{i\mathbf{K}'\cdot(x,y)}v'(x,y),\label{eq:wfB}
\end{equation}
with components $(u,v)$ coming from the approximation close to $\mathbf{K}$
point and fulfilling the system of the two first equations in (\ref{eq:DiractotINIT})
and components $(u',v')$ coming from the approximation close to $\mathbf{K}'$
point and fulfilling the system of the two last equations in (\ref{eq:DiractotINIT}).
An armchair nanoribbon is modelled as a strip $\Pi=(0,L)\times\mathbb{R}$,
$L>0$, parallel to the x-axis. The wave function has to disappear
on the nanoribbon edges, which in the armchair case contain sites from
both sublattices A and B. Consequently it is required \cite{Brey,Macucci}:

\begin{align*}
\psi_{A}(x,0) & =0,\;\;\;\psi_{B}(x,0)=0,\\
\psi_{A}(x,L) & =0,\;\;\;\psi_{B}(x,L)=0.
\end{align*}
From (\ref{eq:wfA}) and (\ref{eq:wfB}), these boundary conditions
transform to

\begin{eqnarray}
u(x,0)-iu'(x,0)=0, &  & -iv(x,0)+v'(x,0)=0,\label{eq:BC1INIT}\\
e^{-i2\pi L}u(x,L)-iu'(x,L)=0, &  & -ie^{-i2\pi L}v(x,L)+v'(x,L)=0.\label{eq:BC2INIT}
\end{eqnarray}
These boundary conditions describe the mixing between valleys which is
characteristic for armchair nanoribbons. For a detailed derivation of
the continuous model see \cite{Macucci}.

Our potential $\mathcal{P}$ is assumed to be of long-range type and
can be described by a diagonal matrix with equal elements \cite{Ando}.

We introduce the energy thresholds
\[
\omega_{1}=\underset{j\in\mathbb{Z}}{\min\Big\{}|\pi+\frac{\pi j}{L}|:|\pi+\frac{\pi j}{L}|>0\Big\},
\]
 and 
\begin{equation}
\omega_{k+1}=\underset{j\in\mathbb{Z}}{\min\Big\{}|\pi+\frac{\pi j}{L}|:|\pi+\frac{\pi j}{L}|>\omega_{k}\Big\},\quad k=1\,,2\,,\ldots.\label{eq:defthreshold}
\end{equation}
Note that 
\begin{equation}
d_{*}\leq|\omega_{k+1}-\omega_{k}|\leq\frac{\pi}{L},\;\;\; d_{*}=\frac{\pi}{L}\min_{\substack{m\in\mathbb{Z}\\
2L+m\neq0
}
}|2L+m|.\label{eq:dstar}
\end{equation}

The continuous spectrum of the problem (\ref{eq:DiractotINIT}), (\ref{eq:BC1INIT}),
(\ref{eq:BC2INIT}) with $\mathcal{P}=0$ depends on the nanoribbon
width $L$ and covers $(-\infty,\omega_{1}]\cup[\omega_{1},+\infty)$.
At the thresholds, the multiplicity of the continuous spectrum changes. 

A trapped mode is defined as a vector eigenfunction (from $L_{2}$
space) that corresponds to an eigenvalue embedded in the continuous
spectrum. The main result of the paper is the following theorem about the
existence of trapped modes in armchair graphene nanoribbons for energies
close to one of the thresholds that can be chosen arbitrary.

\begin{theorem}\label{thrm1}For every $N=1,\,2,\ldots$, there exists
$\epsilon_{N}>0$ such that for each $\epsilon\in(0,\epsilon_{N})$
there exists $\delta\sim\sqrt{\epsilon}$ and a potential $\mathcal{P}$
with $\sup|\mathcal{P}|<1$, such that the problem {\rm (\ref{eq:DiractotINIT})},
{\rm (\ref{eq:BC1INIT})}, {\rm (\ref{eq:BC2INIT})} has a trapped mode
for $\omega=\omega_{N}-\epsilon$.
\end{theorem}

The second result shows that trapped modes may appear only for energies
slightly smaller than a threshold and that the spectrum far from the threshold
is free of embedded eigenvalues, provided the potential is sufficiently small. Moreover their multiplicity does
not exceed one.

\begin{theorem}\label{thrm2_intro}There exist positive numbers $\epsilon_{0}$
independent of $N$ and $\delta_{0}$, such that if 

(i) $\omega\in[\omega_{N},\omega_{N}+\epsilon_{0}]$ or $|\omega-\omega_{k}|>\epsilon_{0}$,
for all $k=1,2,\ldots$ and $|\delta|<\delta_{0}$ with $\delta_{0}$
independent of $N$ and $k$, then the problem {\rm (\ref{eq:DiractotINIT})},
{\rm (\ref{eq:BC1INIT})}, {\rm (\ref{eq:BC2INIT})} has no trapped modes;

(ii) $\omega\in[\omega_{N}-\epsilon_{0},\omega_{N}]$ and $|\delta|<\delta_{0}$
with $\delta_{0}$ which may depend on $N$, then the multiplicity
of a trapped mode to probelm {\rm (\ref{eq:DiractotINIT})}, {\rm (\ref{eq:BC1INIT})},
{\rm (\ref{eq:BC2INIT})} does not exceed $1$.

\end{theorem}

To approach the problem, we follow the technique based on the augumented
scattering matrix developed in \cite{naKamotsky,Nazarov1,Nazarov2,Nazarov3}.

This is the second paper about trapped modes in graphene nanoribbons,
that we consider. In the first one \cite{TrappedArXiv}, we analysed
the case of zigzag nanoribbon with a corresponding non-elliptic boundary
value problem. 

The paper is organised as follows. In Sect. \ref{sec:Dirac-equation}
we analyse the Dirac equation without potential. For a fixed energy, we
construct the bounded solutions (waves) in Sect. \ref{sub:Solutions-with-real}.
Additionally, when the considered energy is close to a threshold (\ref{eq:defthreshold}),
we construct the two unbounded solutions in Sect. \ref{sub:Solutions-with-complex}.
In Sect. \ref{sub:Symplectic-form}, we introduce a symplectic form,
used to define the direction of wave propagation in Sect. \ref{sub:Biorthogonality-real},
\ref{sub:Biorthogonality-complex}. In Sect. \ref{sub:Non-hom}, we
give a solvability result for the non-homogenous problem. 

In Sect. \ref{sec:Dirac-withP}, we include a potential in the Dirac
equation and consider a scattering problem using the augumented
scattering matrix (Sect. \ref{sub:Augumented-scattering-matrix}).
In Sect. \ref{sub:Necessary-and-sufficient}, we give a necessary
and sufficient condition for the existence of trapped modes from which,
in Sect. \ref{sub:Proof-of-Theorem2}, we extract an example potential
that produces a trapped mode and prove Theorem \ref{thrm1}. Finally
in Sec. \ref{sub:Proof-of-Theorem2}, we prove Theorem \ref{thrm2_intro}
about the multiplicity of trapped modes.

\section{The Dirac equation\label{sec:Dirac-equation} }

\subsection{Problem statement\label{sub:Problem-statement}}

First, we consider problem (\ref{eq:DiractotINIT}) without potential
($\mathcal{P}=0$), i.e.
\begin{equation}
\mathcal{D}\left(\begin{array}{c}
u\\
v\\
u'\\
v'
\end{array}\right)=\omega\left(\begin{array}{c}
u\\
v\\
u'\\
v'
\end{array}\right),\label{eq:DiractotnoP}
\end{equation}
with the boundary conditions (\ref{eq:BC1INIT}), (\ref{eq:BC2INIT}).

Our goal is to find solutions of at most exponential growth to the
above problem; in particular, we need bounded solutions to describe the
continuous spectrum of the operator corresponding to (\ref{eq:DiractotnoP}),
(\ref{eq:BC1INIT}), (\ref{eq:BC2INIT}).

Let us introduce the space $X_{0}$ which contains $(u,v,u',v')$,
such that each component belongs to $L^{2}(\Pi)$ and $(i\partial_{x}-\partial_{y})u$,
$(i\partial_{x}+\partial_{y})v$, $(i\partial_{x}+\partial_{y})u'$,
$(-i\partial_{x}+\partial_{y})v'$ are also in $L^{2}(\Pi)$; moreover
the components fulfill the conditions (\ref{eq:BC1INIT}), (\ref{eq:BC2INIT}).
The norm in $X_{0}$ is the usual $L_{2}$-norm for all the components
and their derivatives described above. The operator $\mathcal{D}$ is
self-adjoint in $(L^{2}(\Pi))^{4}$ with the domain $X_{0}$. 
One can verify that the problem (\ref{eq:DiractotnoP}),
(\ref{eq:BC1INIT}), (\ref{eq:BC2INIT}) is elliptic. Another equivalent norm in $X_0$ is given by the following proposition. 
\begin{prop}It holds that 
\[
\int_{\Pi}|\mathcal{D}(u,v,u',v')^{\top}|^{2}dxdy=\int_{\Pi}\Big(|\nabla u|^{2}+|\nabla v|^{2}+|\nabla u'|^{2}+|\nabla v'|^2\Big)dxdy,
\]
for $(u,v,u',v')\in X_0$.
\end{prop}\label{PropEllip}
\begin{proof}
Proof is presented in Appendix \ref{sec:Ellipticity}.
\end{proof}

Due to the nanoribbon geometry we are looking for non-trivial exponential
(or power exponential) solutions in $x$, namely

\begin{equation}
(u(x,y),v(x,y),u'(x,y),v'(x,y))=e^{i\lambda x}(\mathcal{U}(y),\mathcal{V}(y),\mathcal{U}'(y),\mathcal{V}'(y)),\label{uosc}
\end{equation}
where a complex number $\lambda$ is the longitudinal component of
the wave vector parallel with the nanoribbon edge. 

For $\omega=0$ we have $(i\partial_{x}+\partial_{y})v=(i\partial_{x}-\partial_{y})u=(-i\partial_{x}+\partial_{y})v'=(-i\partial_{x}-\partial_{y})u'=0$,
therefore $u=u(x+iy),\, v=v(-x+iy),\, u'=u'(-x+iy),\, v'=v'(x+iy),$
which together with the boundary conditions (\ref{eq:BC1INIT}), (\ref{eq:BC2INIT})
and the form (\ref{uosc}) give solutions only for nanoribbons
with a width $L$ such that $L$ is a natural number and $\lambda=0$. These solutions are
\begin{equation*}
(u,v,u',v')=(1,C,-i,iC),\,\,\, C\in\mathbb{C}.
\end{equation*}
Now assume that $\omega\neq0$, then (\ref{eq:DiractotnoP}) can be
written as:
\begin{equation}
-\Delta u=\omega^{2}u\ ,\ \ v=\omega^{-1}(i\partial_{x}-\partial_{y})u\ ,\ \ -\Delta u'=\omega^{2}u'\ ,\ \ v'=\omega^{-1}(-i\partial_{x}-\partial_{y})u'.\label{delt}
\end{equation}
Then the insertion of (\ref{uosc}) into (\ref{delt}) and (\ref{eq:BC1INIT}), (\ref{eq:BC2INIT}), gives
\begin{equation}
\begin{cases}
-\mathcal{U}_{yy}=(\omega^{2}-\lambda^{2})\mathcal{U}\ ,\ \  & -\mathcal{U}'_{yy}=(\omega^{2}-\lambda^{2})\mathcal{U}'\\
\mathcal{V}=\omega^{-1}(-\lambda\mathcal{U}-\mathcal{U}_{y})\ ,\ \  & \mathcal{V}'=\omega^{-1}(\lambda\mathcal{U}'-\mathcal{U}_{y}'),\\
\mathcal{U}(0)-i\mathcal{U}'(0)=0\ ,\ \  & -i\mathcal{V}(0)+\mathcal{V}'(0)=0\\
e^{-i2\pi L}\mathcal{U}(L)-i\mathcal{U}'(L)=0\ ,\ \  & -ie^{-i2\pi L}\mathcal{V}(L)+\mathcal{V}'(L)=0.
\end{cases}\label{eq:sep}
\end{equation}
If $\lambda^{2}=\omega^{2}$, then problem (\ref{eq:sep}) has a non-trivial
solution only when $L$ is a natural number. In this case 
\begin{equation}
(u(x,y),v(x,y),u'(x,y),v'(x,y))=e^{\pm i\omega x}(1,\mp1,-i,\mp i)\label{eq:solkapzero}
\end{equation}
and there is no power exponential solution.

Now, consider the case $\lambda^{2}\neq\omega^{2}$ ($\omega\neq0$).
If $\omega\neq\omega_{N}$, $N=1\,,2\,,\ldots$, then all the solutions
of (\ref{eq:DiractotnoP}), (\ref{eq:BC1INIT}), (\ref{eq:BC2INIT})
have the form (\ref{uosc}) with 
\begin{equation}
(\mathcal{U}(y),\mathcal{V}(y),\mathcal{U}{}^{'}(y),\mathcal{V}{}^{'}(y))=(e^{i\kappa y},-\omega^{-1}(\lambda+i\kappa)e^{i\kappa y},-ie^{-i\kappa y},-i\omega^{-1}(\lambda+i\kappa)e^{-i\kappa y}),\label{eq:solUV-2-1}
\end{equation}
where 
\begin{equation}
\lambda=\pm\sqrt{\omega^{2}-\kappa^{2}},\,\,\,\kappa=\pi+\frac{\pi j}{L},\:\:\: j=0,\,\pm1,\,\pm2,\ldots.\label{eq:lamsqrt}
\end{equation}
When $\omega=\omega_{N}$, $N=1\,,2\,,\ldots$ and $\lambda=0$, we
have two more solutions when $2L$ is not a natural number
and they are 

\begin{equation}
w_{N}^{0}(x,y)=(e^{i\kappa y},-i\text{sgn}(\kappa)e^{i\kappa y},-ie^{-i\kappa y},\text{sgn}(\kappa)e^{-i\kappa y}),\label{eq:w0}
\end{equation}

\begin{equation}
w_{N}^{1}(x,y)=xw_{N}^{0}(x,y)+|\kappa|^{-1}(0,ie^{i\kappa y},0,-e^{-i\kappa y}),\label{eq:w1}
\end{equation}
with $\kappa=\pi+\frac{\pi m}{L}$, where $m$ is defined through
$\omega_{N}$, 
\[
\omega_{N}=|\pi+\frac{\pi m}{L}|=\underset{j\in\mathbb{Z}}{\min\Big\{}|\pi+\frac{\pi j}{L}|:|\pi+\frac{\pi j}{L}|>\omega_{k}\Big\},\;\; k=1\,,2\,,\ldots,N-1.
\]
It follows that $\kappa$ is one of two $\omega_{N}$ or $-\omega_{N}$.
Differently if $2L$ is a natural number, then we have four more solutions,
two of the form (\ref{eq:w0}) and (\ref{eq:w1}) with $\kappa=\omega_{N}$
and two of the same form (\ref{eq:w0}) and (\ref{eq:w1}) but with
$\kappa=-\omega_{N}$.

\subsection{Symmetries}

There are three symmetries in the system, let's denote them by $T_{1}$,
$T_{2}$ and $T_{3}$. If $(u(x,y),v(x,y),u'(x,y),v'(x,y))$ is a
solution to (\ref{eq:DiractotnoP}), (\ref{eq:BC1INIT}), (\ref{eq:BC2INIT})
then through symmetry transformations $T_{1}$, $T_{2}$ and $T_{3}$,
there are three more solutions which respectively read

\[
\left(\begin{array}{c}
\overline{u}'(x,y)\\
\overline{v}'(x,y)\\
\overline{u}(x,y)\\
\overline{v}(x,y)
\end{array}\right),\left(\begin{array}{c}
v(-x,y)\\
-u(-x,y)\\
-v'(-x,y)\\
u'(-x,y))
\end{array}\right),\left(\begin{array}{c}
e^{i2\pi L}\overline{u}(x,L-y)\\
-e^{i2\pi L}\overline{v}(x,L-y)\\
-\overline{u}'(x,L-y)\\
\overline{v}'(x,L-y)
\end{array}\right).
\]
Superpositions of those symmetries give eight solutions in total in
this case. Moreover, if the nanoribbon width $2L$ is a natural number,
then there is an additional symmetry $T_{4}$ giving the following
solution
\[
\left(\begin{array}{c}
\overline{u}(-x,y)\\
\overline{v}(-x,y)\\
-\overline{u}'(-x,y)\\
-\overline{v}'(-x,y)
\end{array}\right),
\]
and so there are 16 solutions in this case. In what follows we find
solutions $(u,v,u',v')$ for positive $\omega$ only. Then the solution
for $-\omega$ is $(u,-v,u',-v')$.

\subsection{Solutions with real wave number $\lambda$\label{sub:Solutions-with-real}}

In this paper we focus on the case when $2L$ is not a natural
number (Figure \ref{fig:disp} (a) and (b)).

Consider first the case when $\omega\in(\omega_{N-1},\omega_{N})$
for certain $N=1\,,2\,,\ldots$ (we assume that $\omega_{0}=0$).
We enumerate the real exponent in (\ref{eq:lamsqrt}) as follows.
Assuming that $\lambda$ in (\ref{eq:lamsqrt}) is real, we have that

\begin{equation}
-\omega<\pi+\frac{\pi j}{L}<\omega\,\,\,\mbox{or}\,\,\,-L(1+\frac{\omega}{\pi})<j<L(\frac{\omega}{\pi}-1).\label{eq:-om+om}
\end{equation}
According to the last inequality, we enumerate $\kappa$ 
\begin{equation*}
-\omega<\kappa_{1}<\kappa_{2}<\ldots<\kappa_{M}<\omega,
\end{equation*}
where $M=M(\omega)=N-1$ is the number of indexes $j$ satysfying (\ref{eq:-om+om}).
For any $\kappa_{j}$, $j=1\,,2\,,\ldots,M$, there are two values
of $\lambda$ in (\ref{eq:lamsqrt}) that define solutions (\ref{eq:solkapzero}),
(\ref{eq:solUV-2-1}), let us denote them $\pm\lambda_{j}$ with $\lambda_{j}=\sqrt{\omega^{2}-\kappa_{j}^{2}}$,
$j=1,\,2,\,\ldots,M$. We can introduce the notation for solutions
(\ref{eq:solUV-2-1})

\begin{figure}
\includegraphics[scale=0.55]{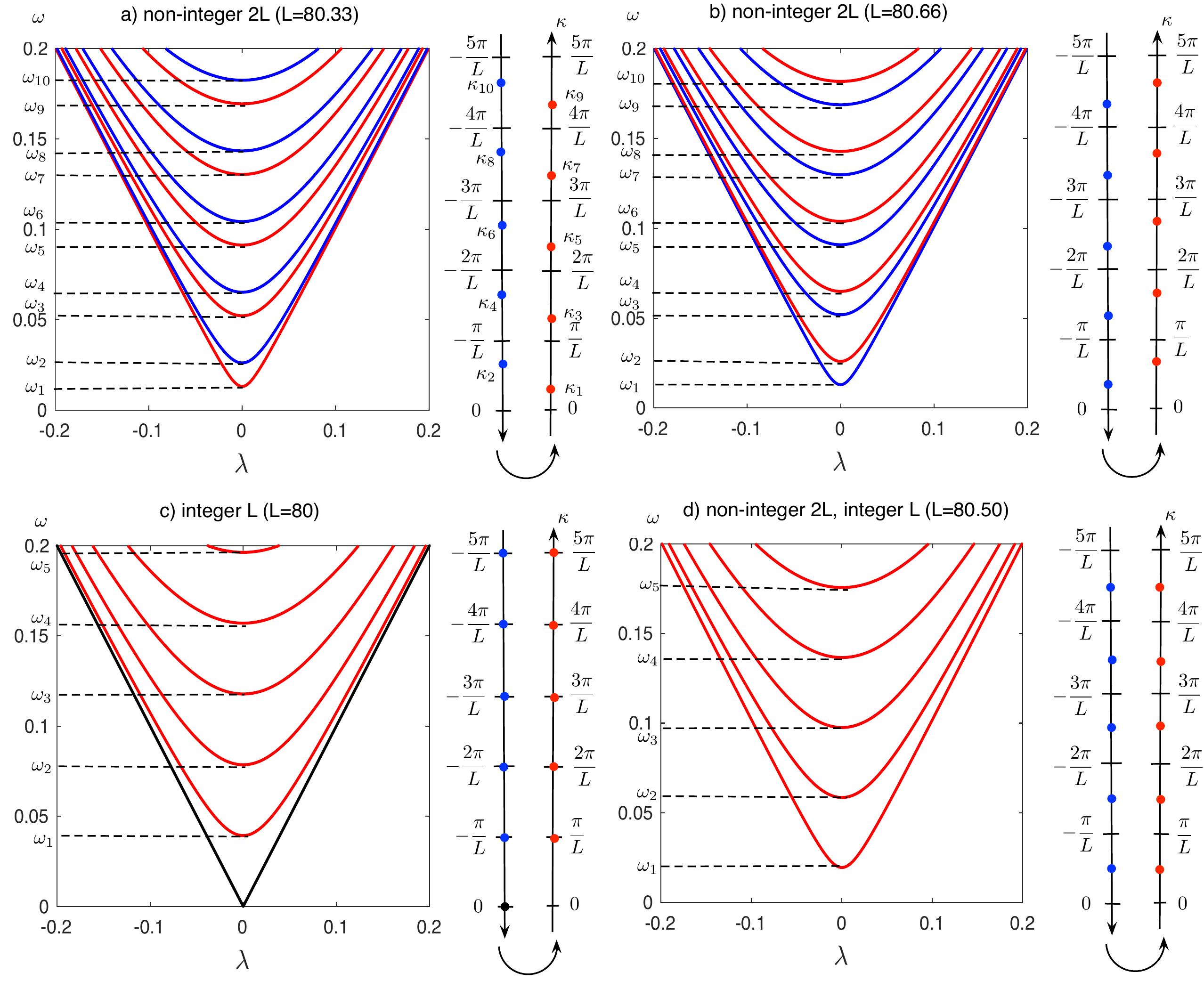}
\caption{{\footnotesize Dispersion relation: energy $\omega$ versus wave vector component
$\lambda$ for nanoribbons of different width $L$: (a) and (b) for
non integer $2L$, (c) integer $L$, (d) integer $2L$ but non integer
$L$. Threshold energies $\omega_{1}<\omega_{2}<\ldots$ and values
of $\kappa$ (blue and red dots) are indicated. The colors of the
dots and curves correspond to the sign of the $\kappa$ value, which when
red is positive and when blue is negative.} \label{fig:disp}}
\end{figure}

\begin{align}
w_{j}^{\pm}(x,y) & =(u_{j}^{\pm}(x,y),v_{j}^{\pm}(x,y),u_{j}{}^{'\pm}(x,y),v_{j}{}^{'\pm}(x,y))\label{eq:soluv}\\
 & =e^{\pm i\lambda_{j}x}(\mathcal{U}_{j}^{\pm}(y),\mathcal{V}_{j}^{\pm}(y),\mathcal{U}{}_{j}^{'\pm}(y),\mathcal{V}{}_{j}^{'\pm}(y)),
\end{align}
where 
\begin{eqnarray*}
 & (\mathcal{U}_{j}^{\pm}(y),\mathcal{V}_{j}^{\pm}(y),\mathcal{U}{}_{j}^{'\pm}(y),\mathcal{V}{}_{j}^{'\pm}(y))\nonumber \\
= & (e^{i\kappa_{j}y},-{\omega}^{-1}(\pm\lambda_{j}+i\kappa_{j})e^{i\kappa_{j}y},-ie^{-i\kappa_{j}y},-i{\omega}^{-1}(\pm\lambda_{j}+i\kappa_{j})e^{-i\kappa_{j}y}).
\end{eqnarray*}

Now, consider the threshold case, that is $\omega=\omega_{N}$, $N=1,\,2,\,,3\,,\ldots$
and $\lambda=0$. Then $\kappa^{2}=\omega^{2}$. In the case $2L$
is not a natural number, there is only one value of $j$ that satisfies
one of the two relations 
\begin{equation*}
\pi+\frac{\pi j}{L}=\pm\omega,
\end{equation*}
using this value of $j$ we define $\kappa_{N}=\pi+\frac{\pi j}{L}$,
with the use of which we find two additional solutions to our problem
(\ref{eq:DiractotnoP}), (\ref{eq:BC1INIT}), (\ref{eq:BC2INIT})

\begin{equation}
w_{N}^{0}(x,y)=(e^{i\kappa_{N}y},-i\text{sgn}(\kappa_{N})e^{i\kappa_{N}y},-ie^{-i\kappa_{N}y},\text{sgn}(\kappa_{N})e^{-i\kappa_{N}y}),\label{eq:w0Ntau}
\end{equation}

\begin{equation}
w_{N}^{1}(x,y)=xw_{N}^{0}(x,y)+{\omega_{N}}^{-1}(0,ie^{i\kappa_{N}y},0,-e^{-i\kappa_{N}y}).\label{eq:w1Ntau}
\end{equation}

Thus, the continuous spectrum depends on the nanoribbon width. For
each $\omega\geq\omega_{1}$ there is a bounded solution to (\ref{eq:DiractotnoP}),
(\ref{eq:BC1INIT}), (\ref{eq:BC2INIT}) of the form (\ref{uosc});
additionally there are bounded solutions for $0<\omega<\omega_{1}$
for $L$ being a natural number. Hence the continuous spectrum for
the Dirac operator $\mathcal{D}$ is $(-\infty,-\omega_{1}]\cup[\omega_{1},\infty)$
when $L$ is not a natural number and $(-\infty,\infty)$ when $L$
is a natural number. Note that $\omega_{1}$ depends on $L$ and is
small for $L$ close to a natural number.

\subsection{Solutions with imaginary wave number $\lambda$\label{sub:Solutions-with-complex}}

Consider the case when $\omega$ is close to the threshold $\omega_{N}$.
Introduce a small parameter $\epsilon$ and denote by $\omega_{\epsilon}$
the energy $\omega_{\epsilon}=\omega_{N}-\epsilon$, $\epsilon>0$.
Then the root $\lambda=0$ bifurcates into two imaginary roots $\pm\lambda_{\epsilon}=\pm\lambda_{\epsilon}(\epsilon)$,
$\lambda_{\epsilon}>0$, which can be found from the equation 
\begin{equation*}
\omega_{\epsilon}^{2}=\omega_{N}^{2}+\lambda_{\epsilon}^{2}.
\end{equation*}
They have the expansion

\begin{equation}
\lambda_{\epsilon}=i\sqrt{\epsilon}\sqrt{2\omega_{N}-\epsilon}),\label{eq:lambdaexpansion}
\end{equation}
so that $\overline{\lambda}_{\epsilon}=-\lambda_{\epsilon}$. 

Now instead of two threshold solutions (\ref{eq:w0Ntau}) and (\ref{eq:w1Ntau}),
we have two solutions of the form (\ref{uosc}) with $ $$\pm\lambda_{\epsilon}$
\begin{equation}
w_{N}^{\pm}(x,y)=e^{\pm i\lambda_{\epsilon}x}\Big(e^{i\kappa_{N}y},-\frac{\pm\lambda_{\epsilon}+i\kappa_{N}}{\omega_{\epsilon}}e^{i\kappa_{N}y},-ie^{-i\kappa_{N}y},-\frac{\pm i\lambda_{\epsilon}-\kappa_{N}}{\omega_{\epsilon}}e^{-i\kappa_{N}y}\Big.\label{eq:wpm}
\end{equation}

By (\ref{eq:lambdaexpansion}), functions (\ref{eq:wpm}) can be written
in terms of $w_{N}^{0}$, $w_{N}^{1}$ (see (\ref{eq:w0}) and (\ref{eq:w1}))
\[
w_{N}^{\pm}(x,y)=w_{N}^{0}(x,y)\mp\sqrt{\epsilon}\sqrt{2\omega_{N}}w_{N}^{1}(x,y)+O(\epsilon).
\]
The waves (\ref{eq:wpm}) are not analytic in $\epsilon$ but their linear
combination
\begin{align}
{\bf w}_{N}^{\epsilon+}(x,y) & =\frac{w_{N}^{+}+w_{N}^{-}}{2}=w_{N}^{0}+O(\epsilon),\label{eq:weps+}
\end{align}
\begin{align}
{\bf w}_{N}^{\epsilon-}(x,y) & =\frac{w_{N}^{+}-w_{N}^{-}}{2\lambda_{\epsilon}}=iw_{N}^{1}+O(\epsilon),\label{eq:weps-}
\end{align}
are analytic with respect to $\epsilon$ for small $|\epsilon|$.

\subsection{Proposition about the location of $\lambda$ }

\begin{figure}
\includegraphics[scale=0.52]{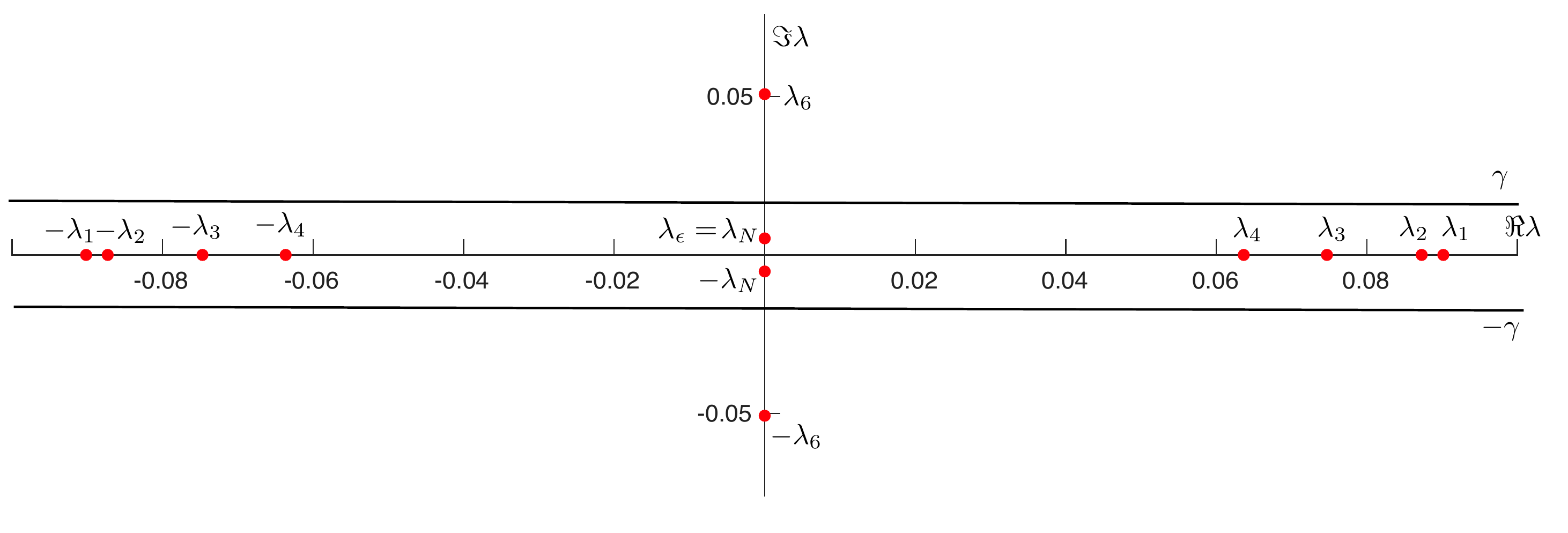}\caption{{\footnotesize The values of $\lambda$ are on the real and imaginary axis only. When the
energy is close to one of the thresholds $\omega_{N}$, it is possible
to choose a strip $|\Im\lambda|\leq\gamma_{N}$ such that all real
and only two imaginary values of $\lambda$ are within the strip (indicated).
For example when $L=80.33$ and $\omega=\omega_{5}-\epsilon$ ($\omega_{5}=0.0912$
and $\epsilon=0.0001$), then all the possible values of $\lambda$
are indicated with red dots, only real values are enumerated: $\pm\lambda_{1},\pm\lambda_{2,},\ldots,\pm\lambda_{4}$
($N=5$) and the imaginary values of $\lambda$ are indicated by $\pm\lambda_{\epsilon}$.
For the sake of the proof of Proposition (\ref{Prop1}), we indicate
$\lambda_{N}$ ($\lambda_{N}=\lambda_{5}$ in our case and $\lambda_{N}=\lambda_{\epsilon}$
as $\omega$ is slightly less and close to the threshold) and $\lambda_{N+1}=\lambda_{6}$.}
$ $\label{fig:lambdaval}}
\end{figure}

\begin{prop}\label{Prop1} There is a constant $\epsilon_{0}=\epsilon_{0}(L)>0$
such that the following assertions are valid

(1) For every $\omega_{N}$ there exist $\gamma_{N}=\gamma_{N}(L)>0$
such that the strip $|\Im\lambda|\leq\gamma_{N}$ contains all the real
and two imaginary values of $\lambda$ described in Sect. {\rm\ref{sub:Solutions-with-real}}
and Sect. {\rm\ref{sub:Solutions-with-complex}} when $\omega\in[\omega_{N}-\epsilon_{0},\omega_{N})$. 

(2) There exist $\gamma=\gamma(L)>0$ such that the strip $|\Im\lambda|\leq\gamma$
contains the real values of $\lambda$ described in Sect. {\rm\ref{sub:Solutions-with-real}}
when $|\omega-\omega_{k}|>\epsilon_{0}$ with $k=1,2,\ldots$
or when $\omega\in[\omega_{N},\omega_{N}+\epsilon_{0}]$.

\end{prop}

\begin{proof}

For the sake of the exposition, let us numerate $\kappa_{N}$ and
$\kappa_{N+1}$ according to (\ref{eq:lamsqrt}) so that $\kappa_{N}=\pi+\frac{\pi m}{L}$,
where $m$ is defined through $\omega_{N}$ 
\[
\omega_{N}=|\pi+\frac{\pi m}{L}|=\underset{j\in\mathbb{Z}}{\min\Big\{}|\pi+\frac{\pi j}{L}|:|\pi+\frac{\pi j}{L}|>\omega_{k}\Big\},\;\; k=1\,,2\,,\ldots,N-1.
\]
It follows that $\kappa_{N}$ is equal to $\omega_{N}$ or $-\omega_{N}$.
Similarly $\kappa_{N+1}=\pi+\frac{\pi m}{L}$ with $m$ defined
through $\omega_{N+1}$, 
\[
\omega_{N+1}=|\pi+\frac{\pi m}{L}|=\underset{j\in\mathbb{Z}}{\min\Big\{}|\pi+\frac{\pi j}{L}|:|\pi+\frac{\pi j}{L}|>\omega_{k}\Big\},\;\; k=1\,,2\,,\ldots,N,
\]
and so $\kappa_{N+1}$ is equal to $\omega_{N+1}$ or $-\omega_{N+1}$.
Consequently, we define $\lambda_{m}=\sqrt{\omega^{2}-\kappa_{m}^{2}}$,
$m=N,\, N+1$ (Figure \ref{fig:lambdaval}).

(1) Let us first estimate the difference $\Im\lambda_{N+1}-\Im\lambda_{N}$.
Let $\epsilon_{0}$ be a small positive number depending on $L$ and
being chosen later. We have 

$ $
\begin{align*}
\Im\lambda_{N+1}-\Im\lambda_{N} & =\sqrt{\omega_{N+1}^{2}-\omega^{2}}-\sqrt{\omega_{N}^{2}-\omega^{2}}=\frac{(\omega_{N+1}^{2}-\omega_{N}^{2})}{\sqrt{\omega_{N+1}^{2}-\omega^{2}}+\sqrt{\omega_{N}^{2}-\omega^{2}}}\\
 & \geq\frac{(\omega_{N+1}-\omega_{N})(\omega_{N+1}+\omega_{N})}{2\sqrt{\omega_{N+1}^{2}-\omega^{2}}}\geq\frac{(\omega_{N+1}-\omega_{N})\sqrt{\omega_{N+1}+\omega_{N}}}{2\sqrt{\omega_{N+1}-\omega_{N}+\epsilon_{0}}}.
\end{align*}
Using estimate (\ref{eq:dstar}), we proceed as 
\[
\Im\lambda_{N+1}-\Im\lambda_{N}\geq\frac{\sqrt{d_{*}}\sqrt{2\omega_{N}}}{2\sqrt{1+\frac{\epsilon_{0}}{\omega_{N+1}-\omega_{N}}}}\geq\frac{\sqrt{d_{*}}\sqrt{2\omega_{N}}}{2\sqrt{1+\frac{\epsilon_{0}}{d_{*}}}}.
\]
Now, $\gamma_{N}$ exists only when $\Im\lambda_{N}$ is small enough.
The upper bound on $\Im\lambda_{N}$ is

\[
\Im\lambda_{N}=\sqrt{\omega_{N}^{2}-\omega^{2}}\leq\sqrt{2\epsilon_{0}\omega_{N}}
\]
and we consider it to be small when $\Im\lambda_{N}<\Im\lambda_{N+1}-\Im\lambda_{N}$
and $\epsilon_{0}$ is chosen so that

\[
\sqrt{\epsilon_{0}2\omega_{N}}\leq\frac{d_{*}\sqrt{2\omega_{N}}}{2\sqrt{d_{*}+\epsilon_{0}}}
\iff
\sqrt{\epsilon_{0}}\leq\frac{d_{*}}{2\sqrt{d_{*}+\epsilon_{0}}}.
\]
From the last inequality, we get that $\epsilon_{0}\leq\frac{1+\sqrt{2}}{2}d_{*}$
and so $\epsilon_{0}$ does not depend on $N$. The value of  $\gamma_{N}$
can be chosen as any between $\Im\lambda_{N}<\gamma_{N}<\Im\lambda_{N+1}-\Im\lambda_{N}$
(Figure \ref{fig:lambdaval}). 

(2) Consider first $|\omega-\omega_{k}|>\epsilon_{0}$ with $k=1,2,\ldots$
. For $\Im\lambda_{k}\neq0$, according to (\ref{eq:lamsqrt}), we
have 
\[
\Im\lambda_{k}=\sqrt{\omega_{k}^{2}-\omega^{2}}\geq\sqrt{\epsilon_{0}}\sqrt{\omega_{k}+\omega}\geq\sqrt{\epsilon_{0}}\sqrt{\omega_{1}}.
\]
Choosing $\gamma=\frac{1}{2}\sqrt{\epsilon_{0}\omega_{1}}$, we get
that the strip contains only real values of $\lambda$.

If $\omega\in[\omega_{N},\omega_{N}+\epsilon_{0}]$, then $\Im\lambda_{N}=0$
and 
\begin{align*}
\Im\lambda_{N+1} & =\sqrt{\omega_{N+1}^{2}-\omega^{2}}\geq\sqrt{\omega_{N+1}-\omega_{N}+\epsilon_{0}}\sqrt{\omega_{N+1}+\omega_{N}}\\
 & \geq\sqrt{d_{*}+\epsilon_{0}}\sqrt{\omega_{2}+\omega_{1}}\geq\sqrt{\epsilon_{0}}\sqrt{\omega_{1}}.
\end{align*}
Again, we can choose $\gamma=\frac{\sqrt{\epsilon_{0}\omega_{1}}}{2}$ so that 
the strip contains only real values of $\lambda$.

\end{proof}

\subsection{The symplectic form\label{sub:Symplectic-form}}

For two solutions $w=(u,v,u',v')$ and $\tilde{w}=(\tilde{u},\tilde{v},\tilde{u}',\tilde{v}')$
of the problem (\ref{eq:DiractotnoP}), (\ref{eq:BC1INIT}), (\ref{eq:BC2INIT}),
let us define the quantity

\begin{equation}
q_{a}(w,\tilde{w})=-i\int_{0}^{L}\overline{\tilde{u}}(a,y)v(a,y)+\overline{\tilde{v}}(a,y)u(a,y)-\overline{\tilde{u}}'(a,y)v'(a,y)-\overline{\tilde{v}}'(a,y)u'(a,y)dy.\label{eq:q}
\end{equation}
Since
\begin{multline*}
0=\int_{\Pi_{a,b}}\left(\begin{array}{c}
\overline{\tilde{u}}\\
\overline{\tilde{v}}\\
\overline{\tilde{u}'}\\
\overline{\tilde{v}'}
\end{array}\right)({\cal D}-\omega{\cal I})\left(\begin{array}{c}
u\\
v\\
u'\\
v'
\end{array}\right)dxdy-\int_{\Pi_{a,b}}\left(\begin{array}{c}
u\\
v\\
u'\\
v'
\end{array}\right)(\overline{\mathcal{D}}-\omega{\cal I})\left(\begin{array}{c}
\overline{\tilde{u}}\\
\overline{\tilde{v}}\\
\overline{\tilde{u}'}\\
\overline{\tilde{v}'}
\end{array}\right)dxdy\\
=-q_{b}(w,\tilde{w})+q_{a}(w,\tilde{w}),
\end{multline*}
where $\Pi_{a,b}=(0,L)\times(a,b)$, $a<b$, we see that $q_{a}$
does not depend on $a$ and we use the notation $q$ for this
form.

The form q is symplectic because it is sesquilinear and anti-Hermitian.

\subsection{Biorthogonality conditions when the wave vector $\lambda$ is real\label{sub:Biorthogonality-real}}

Here we discuss the biorthogonality conditions for the solutions to (\ref{eq:DiractotnoP}),
(\ref{eq:BC1INIT}), (\ref{eq:BC2INIT}). Since we are interested
mostly in the case when $\omega=\omega_{\epsilon}$, where $\omega_{\epsilon}=\omega_{N}-\epsilon$,
we consider this case here. Using (\ref{eq:q}), we obtain the biorthogonality
conditions for the oscillatory waves $w_{j}^{\pm}$ in (\ref{eq:soluv}):
\begin{equation*}
q(w_{j}^{\tau},w_{k}^{\theta})=0\;\;\mbox{if \ensuremath{(j,\tau)\neq(k,\theta)}\;\;\mbox{and \ensuremath{q(w_{j}^{\tau},w_{j}^{\tau})=\frac{4\tau iL\lambda_{j}}{\omega}}}.}
\end{equation*}
Therefore 
\begin{equation}
q(w_{j}^{\tau},w_{k}^{\theta})=\frac{4\tau iL\lambda_{j}}{\omega}\delta_{j,k}\delta_{\tau,\theta},\label{April1a}
\end{equation}
for $j,k=1,\ldots,N-1$ and $\tau,\theta=\pm$. We put 
\begin{equation*}
{\bf w}_{k}^{\tau}=\frac{\sqrt{\omega}}{2\sqrt{L|\lambda_{k}|}}w_{k}^{\tau},\;\;\; k=1,\ldots,N-1,\;\;\;\tau=\pm.
\end{equation*}
Then by (\ref{April1a}), we have:
\begin{equation}
q({\bf w}_{j}^{\tau},{\bf w}_{k}^{\theta})=\tau i\delta_{j,k}\delta_{\tau,\theta}.\label{eq:bioosc}
\end{equation}
and when $\omega=\omega_{N}$, then
\begin{equation}
q(w_{N}^{0},w_{N}^{0})=0,\ q(w_{N}^{0},w_{N}^{1})=-\frac{2L}{\omega_{N}},\; q(w_{N}^{1},w_{N}^{1})=0.
\end{equation}

\subsection{Biorthogonality conditions when the wave vector $\lambda$ is imaginary\label{sub:Biorthogonality-complex}}

Let us check if the waves $\mathbf{w}_{N}^{\epsilon\pm}$ fulfil the
orthogonality conditions. A direct evaluation gives
\begin{equation*}
q(w_{N}^{\pm},w_{N}^{\pm})=0,\ \ \ q(w_{N}^{-},w_{N}^{+})=\frac{i4L}{\omega_{\epsilon}}\lambda_{+},\ \ \ q(w_{N}^{+},w_{N}^{-})=\frac{i4L}{\omega_{\epsilon}}\lambda_{-}.
\end{equation*}
 Consequently
\begin{equation*}
q({\bf w}_{N}^{\epsilon+},{\bf w}_{N}^{\epsilon+})=0,\ \ \ q({\bf w}_{N}^{\epsilon+},{\bf w}_{N}^{\epsilon-})=\frac{i2L}{\omega_{\epsilon}},\ \ \ q({\bf w}_{N}^{\epsilon-},{\bf w}_{N}^{\epsilon-})=0.
\end{equation*}
As waves ${\bf w}_{N}^{\epsilon+}$ and ${\bf w}_{N}^{\epsilon-}$
do not fulfil the biorthogonality conditions, we introduce their linear
combinations
\begin{equation}
{\bf w}_{N}^{\pm}=\frac{{\bf w}_{N}^{\epsilon+}\pm{\bf w}_{N}^{\epsilon-}}{\mathcal{N}}=\frac{1}{2\mathcal{N}}\Big(1\pm\frac{1}{\lambda_{\epsilon}}\Big)w_{N}^{+}+\frac{1}{2\mathcal{N}}\Big(1\mp\frac{1}{\lambda_{\epsilon}}\Big)w_{N}^{-},\;\;\;\mathcal{N}=2\sqrt{\frac{L}{\omega_{\epsilon}}}.\label{eq:WNPWNM}
\end{equation}
 Then the new waves (\ref{eq:WNPWNM}) fulfill the condition

\[
q({\bf w}_{N}^{\tau},{\bf w}_{N}^{\theta})=\tau i\delta_{\tau,\theta}.
\]

\subsection{The non-homogeneous problem\label{sub:Non-hom}}

Consider the non-homogeneous problem 
\begin{equation}
(i\partial_{x}+\partial_{y})v-\omega u=g\;\;\;\mbox{in \ensuremath{\Pi},}\label{1}
\end{equation}
\begin{equation}
(i\partial_{x}-\partial_{y})u-\omega v=h\;\;\;\mbox{in\,\ \ensuremath{\Pi},}\label{2}
\end{equation}
\begin{equation}
(-i\partial_{x}+\partial_{y})v'-\omega u'=g'\;\;\;\mbox{in \ensuremath{\Pi},}\label{3}
\end{equation}
\begin{equation}
(-i\partial_{x}-\partial_{y})u'-\omega v'=h'\;\;\;\mbox{in\,\ \ensuremath{\Pi},}\label{4}
\end{equation}
supplied with boundary conditions (\ref{eq:BC1INIT}), (\ref{eq:BC2INIT}). 

In order to formulate the solvability results for this problem, we introduce
some spaces. The space $L_{\sigma}^{\pm}(\Pi)$, $\sigma>0$, consists
of all functions $g$ such that $e^{\pm\sigma x}g\in L^{2}(\Pi)$.
Then the space $X_{\sigma}^{\pm}$ contains $(u,v,u',v')$ such that
$e^{\pm\sigma x}(u,v,u',v')\in X_{0}$. The norms in the above spaces
are defined by $||g;L_{\sigma}^{\pm}(\Pi)||=||e^{\pm\sigma x}g;L^{2}(\Pi)||$
and $||(u,v,u',v');X_{\sigma}^{\pm}||=||e^{\pm\sigma x}(u,v,u',v');X_{0}||$
respectively.

\begin{theorem}\label{T1} Let $\omega>0$ and let $\sigma>0$ be
such that the line $\Im\lambda=\pm\sigma$ contains no $\lambda$
defined by {\rm (\ref{eq:lamsqrt})}. Then the operator%
\footnote{For the simplicity of the notation, we write $L_{\sigma}^{\pm}(\Pi)$
for both spaces of functions and vectors. Here for example we write
$L_{\sigma}^{\pm}(\Pi)$ instead of $L_{\sigma}^{\pm}(\Pi)\times L_{\sigma}^{\pm}(\Pi)\times L_{\sigma}^{\pm}(\Pi)\times L_{\sigma}^{\pm}(\Pi)$.
This notation is applied to the other spaces introduced later
as well. %
} 
\begin{equation*}
\mathcal{D}-\omega{\cal I}\;:\; X_{\sigma}^{\pm}\;\to\; L_{\sigma}^{\pm}(\Pi)
\end{equation*}
is an isomorphism.\end{theorem} 

\begin{proof}
The following result is a consequence of ellipticity and it follows
from Theorem 2.4.1 in \cite{KozlovBV}. To apply Theorem 2.4.1 in \cite{KozlovBV}, we put that $l=1$, $H_0=L_2(0,L)^4$ and $H_1=\{(\mathcal{U},\mathcal{V},\mathcal{U}',\mathcal{V}')\in H^1(0,L)^4:\ \mathcal{U}(0)-i\mathcal{U}'(0)=0,\  -i\mathcal{V}(0)+\mathcal{V}'(0)=0,\ 
e^{-i2\pi L}\mathcal{U}(L)-i\mathcal{U}'(L)=0,\ -ie^{-i2\pi L}\mathcal{V}(L)+\mathcal{V}'(L)=0\}$, then by Prop. \ref{PropEllip}, Condition I on p.27 and Condition II on p.28 in \cite{KozlovBV} are fulfilled and we can apply Theorem 2.4.1 in \cite{KozlovBV}. The assertion of this theorem can be obtained also from Theorem 1.1 in \cite{NaPl}.
\end{proof}

In what follows we assume that the integer $N$ defining a threshold
$\omega_{N}$ is fixed. Then according to Proposition \ref{Prop1},
there exist $\epsilon_{0}$ and $\gamma_{N}$ such that for $\omega=\omega_{\epsilon}$
with some small positive $\epsilon\in[0,\epsilon_{0}]$, the strip
$|\Im\lambda|\leq\gamma_{N}$ contains only real wavenumbers $\pm\lambda_{j}$,
$j=1,\ldots,N-1$ and two imaginary $\pm\lambda_{\epsilon}$ and the corresponding
wavefunctions are

\[
{\bf w}_{1}^{\pm},{\bf w}_{2}^{\pm},\ldots,{\bf w}_{N-1}^{\pm},{\bf w}_{N}^{\pm}.
\]
All waves but last two are oscillatory. Those last waves are of exponential
growth. 

\begin{theorem}\label{T1s} Let $\gamma_{N}$ and $\epsilon_{0}$
be the same positive numbers as in Proposition \ref{Prop1} and let
also $(g,h,g',h')\in L_{\gamma}^{+}(\Pi)\cap L_{\gamma}^{-}(\Pi)$.
Denote by $(u^{\pm},v^{\pm},u{}^{'\pm},v{}^{'\pm})\in X_{\gamma}^{\pm}$
the solution of problem {\rm(\ref{1})}, {\rm (\ref{2})}, {\rm(\ref{3})}, {\rm(\ref{4})}
with the boundary conditions {\rm(\ref{eq:BC1INIT})}, {\rm(\ref{eq:BC2INIT})},
which exist according to Theorem {\rm\ref{T1}}. Then 
\begin{equation*}
(u^{+},v^{+})=(u^{-},v^{-})+\sum_{j=1}^{N}C_{j}^{+}{\bf w}_{j}^{+}+\sum_{j=1}^{N}C_{j}^{-}{\bf w}_{j}^{-},
\end{equation*}
where 
\begin{equation*}
-iC_{j}^{+}=\int_{\Pi}\,(g,h,g',h')\cdot\overline{{\bf w}_{j}^{+}}dxdy,\;\; iC_{j}^{-}=\int_{\Pi}\,(g,h,g',h')\cdot\overline{{\bf w}_{j}^{-}}dxdy.
\end{equation*}
\end{theorem} 
\begin{proof}
The functions ${\bf w}_{j}^{\pm}$ correspond to functions (2.11) and (2.12) on p.30 from \cite{KozlovBV}, and so by Prop. 2.8.1 in \cite{KozlovBV} we obtain the claim in the theorem.
\end{proof}

\section{The Dirac equation with potential\label{sec:Dirac-withP}}

\subsection{Problem statement}

Here we examine the problem with a potential, prove its solvability result
and asymptotic formulas for the solutions. Consider the nanoribbon with
a potential: 
\begin{equation}
\mathcal{D}\left(\begin{array}{c}
u\\
v\\
u'\\
v'
\end{array}\right)+\delta\mathcal{P}\left(\begin{array}{c}
u\\
v\\
u'\\
v'
\end{array}\right)=\omega\left(\begin{array}{c}
u\\
v\\
u'\\
v'
\end{array}\right),\label{Dirac}
\end{equation}
with the boundary conditions (\ref{eq:BC1INIT}), (\ref{eq:BC2INIT}).
Here, ${\cal D}$ is defined in (\ref{def:dirac}), $\mathcal{P}=\mathcal{P}(x,y)$
is a bounded, continuous real-valued function with compact support
in $\overline{\Pi}$ and $\delta$ is a small parameter. We assume
in what follows that 
\begin{equation*}
\mbox{supp}\mathcal{P}\subset[-R_{0},R_{0}]\times[0,1]\:\:\:\mbox{and}\:\:\:\sup_{(x,y)\in\Pi}|\mathcal{P}|\leq1,
\end{equation*}
where $R_{0}$ is a fixed positive number.

We assume that $N$, $\gamma$, $\epsilon_{0}$ are fixed and
\begin{equation*}
\omega=\omega_{\epsilon},\;\;\;\mbox{where}\;\;\;\omega_{\epsilon}=\omega_{N}-\epsilon,
\end{equation*}
with $\epsilon\in[0,\epsilon_{0}]$ according to \ref{Prop1}(i).

Since the norm of the multiplication by $\delta\mathcal{P}$ operator
in $L^{2}(\Pi)$ is less than $\delta$ we derive from Theorem \ref{T1}
the following

\begin{theorem}\label{T1a} The operator 
\begin{equation*}
{\cal D}+(\delta\mathcal{P}-\omega_{\epsilon}){\cal I}\;:\; X_{\gamma}^{\pm}\;\to\; L_{\gamma}^{\pm}(\Pi)
\end{equation*}
is an isomorphism for $|\delta|\leq\delta_{0}$, where $\delta_{0}$
is a positive constant depending on the norm on the inverse operator
$({\cal D}-\omega_{\epsilon}I)^{-1}\,:\, L_{\gamma}^{\pm}(\Pi)\to L_{\gamma}^{\pm}(\Pi)$.
\end{theorem}

We introduce two new spaces for $\gamma>0$ 
\[
{\cal H}_{\gamma}^{+}=\{(u,v,u',v')\;:\;(u,v,u',v')\in X_{\gamma}^{+}\cap X_{\gamma}^{-}\}
\]
and 
\[
{\cal H}_{\gamma}^{-}=\{(u,v,u',v')\;:\;(u,v,u',v')\in X_{\gamma}^{+}\cup X_{\gamma}^{-}\}.
\]
The norms in this spaces are defined by 
\[
||(u,v,u',v');{\cal H}_{\gamma}^{\pm}||^{2}=\int_{\Pi}e^{\pm2\gamma|x|}\Big(|u|^{2}+|v|^{2}+|u'|^{2}+|v'|^{2}+|{\cal D}(u,v,u',v')^{t}|^{2}\Big)dxdy.
\]
Note that, $\mathcal{H}_{0}^{+}=\mathcal{H}_{0}^{-}=X_{0}$, where
$X_{0}$ was introduced in Sect. \ref{sub:Problem-statement}.

Let also ${\cal L}_{\gamma}^{\pm},$ $\gamma>0$, be two weighted $L^{2}$-spaces in $\Pi$ with the norms 
\[
||(u,v,u',v');{\cal L}_{\gamma}^{\pm}||^{2}=\int_{\Pi}e^{\pm2\gamma|x|}\Big(|u|^{2}+|v|^{2}+|u'|^{2}+|v'|^{2}\Big)dxdy.
\]
We define two operators acting in the introduced spaces 
\begin{equation*}
A_{\gamma}^{\pm}=A_{\gamma}^{\pm}(\epsilon,\delta)={\cal D}+(\delta\mathcal{P}-\omega){\cal I}\,:\,{\cal H}_{\gamma}^{\pm}\rightarrow{\cal L}_{\gamma}^{\pm}.
\end{equation*}
Some important properties of these operator are collected in the following

\begin{theorem}\label{T3s} The operators $A_{\gamma}^{\pm}$ are
Fredholm and ${\rm ker}A_{\gamma}^{+}=\{0\}$, coker$A_{\gamma}^{-}=\{0\}$.
Moreover 
\[
{\rm dim}\,{\rm coker}A_{\gamma}^{+}={\rm dim}\,{\rm ker}A_{\gamma}^{-}=2N.
\]
 \end{theorem}

\begin{proof} The proof repeats reasoning presented in the proof of
Theorem 3.2 in \cite{TrappedArXiv}.

\end{proof}

\textcolor{black}{In the next theorem and in what follows, we fix
four smooth functions, $\chi_{\pm}=\chi_{\pm}(x)$ and $\eta_{\pm}=\eta_{\pm}(x)$
such that $\chi_{+}(x)=1$, $\chi_{-}(x)=0$ for $x>R_{0}$ and $\chi_{+}(x)=0$,
$\chi_{-}(x)=1$ for $x<-R_{0}$. Then let $\eta_{\pm}(x)=1$ for
large positive $\pm x$, $\eta_{\pm}(x)=0$ for large negative $\pm x$
and $\chi_{\pm}\eta_{\pm}=\chi_{\pm}$.}

Let us derive an asymptotic formula for the solution to the perturbed
problem (\ref{Dirac}), (\ref{eq:BC1INIT}), (\ref{eq:BC2INIT}).

\begin{theorem}\label{T3s2} Let $f\in{\cal L}_{\gamma}^{+}$ and
let $w=(u,v,u',v')\in{\cal H}_{\gamma}^{-}$ be a solution to

\begin{equation}
({\cal D}+(\delta\mathcal{P}-\omega){\cal I})w=f.\label{April9b}
\end{equation}
satisfying {\rm(\ref{eq:BC1INIT})}, {\rm(\ref{eq:BC2INIT})}.Then 
\begin{equation*}
w=\eta_{+}\sum_{j=1}^{N}\sum_{\tau=\pm}C_{j}^{\tau}{\bf w}_{j}^{\tau}+\eta_{-}\sum_{j=1}^{N}\sum_{\tau=\pm}D_{j}^{\tau}{\bf w}_{j}^{\tau}+R,
\end{equation*}
where $R\in{\cal H}_{\gamma}^{+}$ . \end{theorem}

\begin{proof} The proof is analogous to the proof of Theorem 3.3
in \cite{TrappedArXiv}.

\end{proof}

\subsection{The augumented scattering matrix\label{sub:Augumented-scattering-matrix}}

The scattering matrix is our main tool for the identification of the trapped
modes. Using the q-form, we define the incoming/outgoing waves. The scattering
matrix is defined via coefficients in this combination of waves. It
is important to point out that this matrix is often called augumented
as it contains coefficients of the waves exponentially growing at
infinity as well. Finally, by the end of the section we define a space
with separated asymptotics and check that it produces a unique solution
to the perturbed problem.

Let
\begin{equation*}
Q_{R}(w,\tilde{w})=q_{R}(w,\tilde{w})-q_{-R}(w,\tilde{w}).
\end{equation*}
If $w=(u,v)$ and $\tilde{w}=(\tilde{u},\tilde{v})$ are solutions
to (\ref{Dirac}), (\ref{eq:BC1INIT}), (\ref{eq:BC2INIT}) for $|y|\geq R_{0}$,
then using the Green's formula one can show that this form is independent
of $R\geq R_{0}$. We introduce two sets of localized waves at $\pm\infty$
waves, which we call outgoing and incoming (for physical interpretation
see Appendix \ref{sec:Mandelstam-radiation-condition}) 
\begin{equation}
W_{k}^{\pm}=W_{k}^{\pm}(x,y;\epsilon)=\chi_{\pm}(y){\bf w}_{k}^{\pm}(x,y)\label{W1a}
\end{equation}
and 
\begin{equation}
V_{k}^{\mp}=V_{k}^{\mp}(x,y;\epsilon)=\chi_{\pm}(y){\bf w}_{k}^{\mp}(x,y)\label{W1b}
\end{equation}
with $k=1,\ldots,N$. The reason for introducing this sets of waves
is the property 
\begin{equation}
Q_{R}(W_{k}^{\tau},W_{j}^{\theta})=i\delta_{k,j}\delta_{\tau,\theta},\;\; Q_{R}(V_{k}^{\tau},V_{j}^{\theta})=-i\delta_{k,j}\delta_{\tau,\theta}\label{W1c}
\end{equation}
where for $j,k=1,\ldots,N$. Moreover, 
\begin{equation}
Q_{R}(W_{k}^{\tau},V_{j}^{\theta})=0.\label{W1ca}
\end{equation}
Thus the sign of the $Q$ product separates waves $W$ and $V$.

In the next lemma we give a description of the kernel of the operator
$A_{\gamma}^{-}$, which is used in the definition of the scattering
matrix.

\begin{theorem}\label{tTh1} There exists a basis in $\ker A_{\gamma}^{-}$
of the form 
\begin{equation}
z_{k}^{\tau}=V_{k}^{\tau}+\sum_{\theta=\pm}\sum_{j\in\mathcal{I}_{\kappa}}{\bf S}_{k\tau}^{j\theta}W_{j}^{\theta}+\tilde{z}_{k}^{\tau},\label{zk}
\end{equation}
where $\tilde{z}_{k}^{\tau}\in{\cal H}_{\gamma}^{+}$. Moreover, the
coefficients ${\bf S}_{k\tau}^{j\theta}={\bf S}_{k\tau}^{j\theta}(\epsilon,\delta)$
are uniquely defined. \end{theorem} 

\begin{proof} The proof repeats the reasoning presented in the proof
of Theorem 3.4 in \cite{TrappedArXiv}.

\end{proof}

The scattering matrix ${\bf S}$ is defined through the kernel of operator
$A_{\gamma}^{-}$ that is through the formula (\ref{zk}).

\subsection{The block notation}

An important role in the construction of a trapped mode, is played by
a part of the scattering matrix defined within the block notation.

Let us write 
\[
{\bf W}=({\bf W}_{\bullet},{\bf W}_{\dagger}),{\bf \;\;\; V}=({\bf V}_{\bullet},{\bf V}_{\dagger}),
\]
where 
\[
{\bf W}_{\bullet}=(W_{1}^{+},W_{1}^{-},\ldots,W_{N-1}^{+},W_{\text{N-1}}^{-}),\;\;{\bf V}_{\bullet}=(V_{1}^{+},V_{1}^{-},\ldots,V_{N-1}^{+},V_{N-1}^{-}),
\]
and
\[
{\bf W}_{\dagger}=(W_{N}^{+},W_{N}^{-}),\:\:\:{\bf V}_{\dagger}=(V_{N}^{+},V_{N}^{-}).
\]

Equation (\ref{zk}) in the vector form reads 
\begin{equation}
{\bf z}={\bf V}+{\bf S}{\bf W}+{\bf r}\label{z_sol}
\end{equation}
with 
\begin{equation*}
{\bf z}=({\bf z}_{\bullet},{\bf z}_{\dagger}),\;\;\;{\bf r}=({\bf r}_{\bullet},{\bf r}_{\dagger}),
\end{equation*}
and
\begin{equation}
{\bf z}_{\bullet}=(z_{1}^{+},z_{1}^{-},..,z_{N-1}^{+},z_{N-1}^{-})\in{\cal H}_{\gamma}^{-},\;\;\;{\bf r}_{\bullet}=(r_{1}^{+},r_{1}^{-},..,r_{N-1}^{+},r_{N-1}^{-})\in{\cal H}_{\gamma}^{+},
\end{equation}

\[
{\bf z}_{\dagger}=(z_{N}^{+},z_{N}^{-})\in{\cal H}_{\gamma}^{-},\;\;\;{\bf r}_{\dagger}=(r_{N}^{+},r_{N}^{-})\in{\cal H}_{\gamma}^{+}.
\]
Here both vectors ${\bf z}$ and $r$ have $2N$ elements. The matrix
${\bf S}={\bf S}(\epsilon,\delta)$ is written in the block form 
\begin{equation*}
\mathcal{S}=\left(\begin{array}{cc}
{\bf S}_{\bullet\bullet} & {\bf S}_{\bullet\dagger}\\
{\bf S}_{\dagger\bullet} & {\bf S}_{\dagger\dagger}
\end{array}\right).
\end{equation*}
The symmetry transformations $T_{1},T_{2},T_{3}$ applied to waves
(\ref{eq:soluv}), (\ref{eq:WNPWNM}) gives 
\[
T_{1}\mathbf{w}_{j}^{\tau}=i\mathbf{w}_{j}^{-\tau},\;\;\; T_{3}\mathbf{w}_{j}^{\tau}=e^{i\kappa_{j}L}\mathbf{w}_{j}^{-\tau},\;\;\; j=1,2,\ldots,N,\;\tau=\pm,
\]
 
\[
T_{2}\mathbf{w}_{j}^{\tau}=-\frac{\tau\lambda_{j}+i\kappa_{j}}{\omega}\mathbf{w}_{j}^{-\tau},\;\;\; j=1,2,\ldots,N-1,\;\tau=\pm
\]
and

\[
T_{2}\mathbf{w}_{N}^{+}=\frac{\lambda_{\epsilon}^{2}-1}{2\omega_{\epsilon}}\mathbf{w}_{N}^{+}+\frac{-2i\kappa_{N}-1-\lambda_{\epsilon}^{2}}{2\omega_{\epsilon}}\mathbf{w}_{N}^{-},
\]
\[
T_{2}\mathbf{w}_{N}^{-}=\frac{-2i\kappa_{N}+1+\lambda_{\epsilon}^{2}}{2\omega_{\epsilon}}\mathbf{w}_{N}^{+}+\frac{-\lambda_{\epsilon}^{2}+1}{2\omega_{\epsilon}}\mathbf{w}_{N}^{-}.
\]
The symmetry $T_{1}$ leads to important properties of the matrix $S$,
that are used later, in Sect. \ref{sub:Proof-of-Theorem1}, namely

\begin{equation*}
S_{k\tau}^{k\tau}=S_{k(-\tau)}^{k(-\tau)},\;\;\; S_{k\tau}^{j\theta}=S_{j(-\theta)}^{k(-\tau)},\;\;\; k,\, j=1,\ldots,N,\;\; k\neq j,\;\;\tau,\,\theta=\pm,
\end{equation*}
that is illustrated in colorful bullets in Figure \ref{fig:matrixS}.
In the case $\mathcal{P}(x,y)=\mathcal{P}(x,L-y)$, the symmetry $T_{3}$
implies 
\[
S_{N\tau}^{j\theta}=0\;\;\;\mbox{for}\;\;\; N-j\;\mbox{being an odd number}
\]
and 
\[
S_{N-}^{k+}=S_{N+}^{k-},\;\;\; S_{k+}^{N-}=S_{k-}^{N+}.
\]

The relations for symmetry $T_{2}$ require the assumption $\mathcal{P}(x,y)=\mathcal{P}(-x,y)$
and are more complicated. In our construction of a potential introducing
a trapped mode in Sect. \ref{sub:Proof-of-Theorem1}, we use the symmetries
$T_{1}$ and $T_{3}$ assuming $\mathcal{P}(x,y)=\mathcal{P}(x,L-y)$;
however we do not require $\mathcal{P}(x,y)=\mathcal{P}(-x,y)$.

Relations (\ref{W1c}) and (\ref{W1ca}) take the form 
\begin{equation}
Q({\bf W},{\bf W})=i\mathbb{I}\ ,\ \ Q({\bf V},{\bf V})=-i\mathbb{I}\ \ \ Q({\bf W},{\bf V})=\mathbb{O}.\label{eq:Qprop}
\end{equation}
where $\mathbb{I}$ is the identity matrix and $\mathbb{O}$ is the
null matrix of appropriate size.

\begin{figure}
\centering{}\includegraphics[scale=0.6]{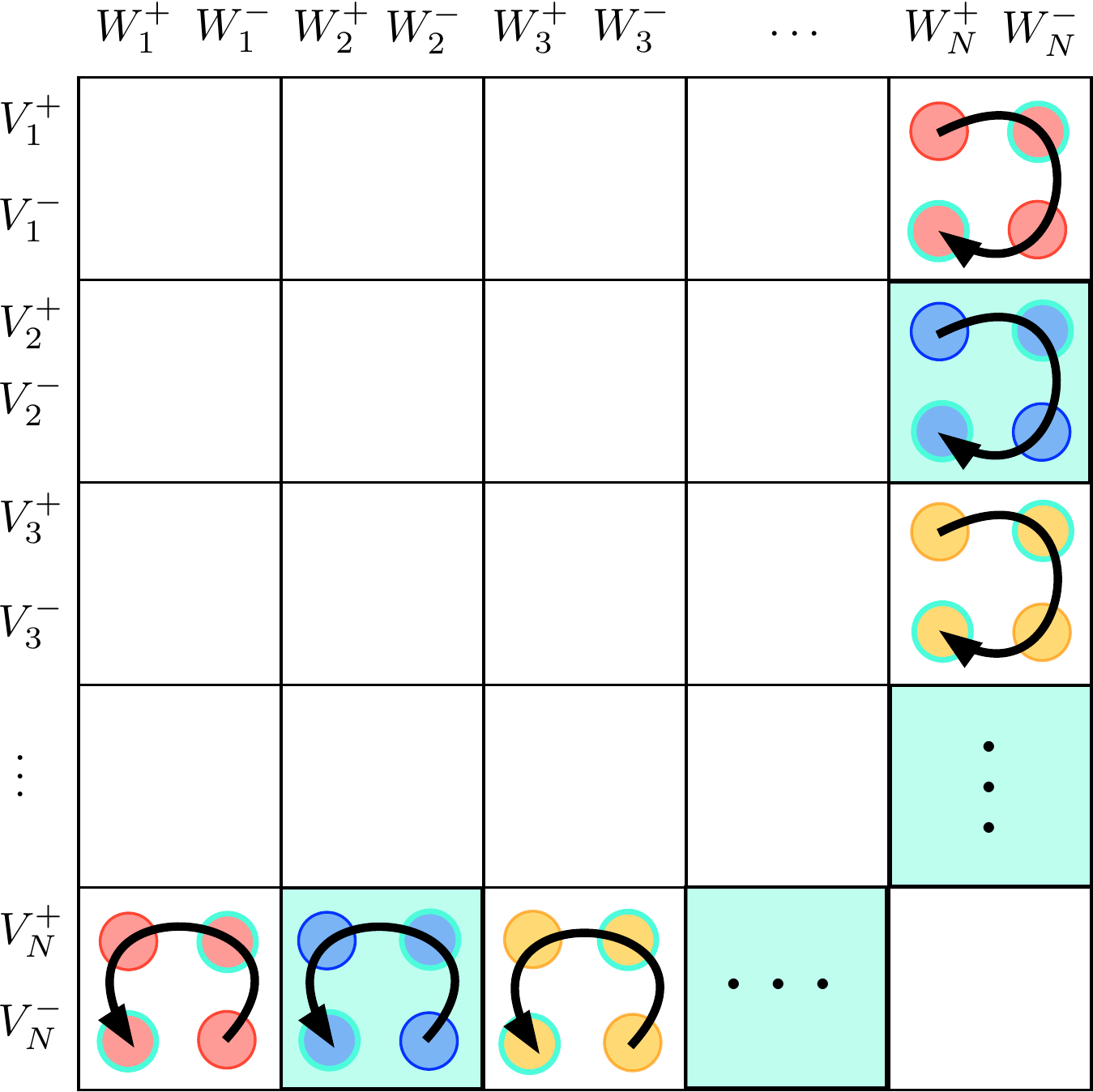}
\caption{{\footnotesize Symmetries in the augumented scattering matrix $S$ (only the part
with $S_{\dagger\bullet}$ and $S_{\bullet\dagger}$ are considered). According to symmetry $T_{1}$,
the left bottom square has permuted elements of the right top square:
they are indicated in red and the permutation follows the arrows,
that is the first element is placed by the beginnings of arrows and
the last by the ends. Similar permutations for other squares are indicated
in blue and yellow. From symmetry $T_{3}$ follows that all the elements
in the green squares (every second square but not the right bottom
square) are zero; moreover elements within square with red elements
that have green borders are equal, similarly for blue and yellow elements.}
\label{fig:matrixS}}
\end{figure}

\subsection{Properties of the scattering matrix}

\begin{prop}\label{sdeltaunitary} The scattering matrix ${\bf S}$
is unitary. \end{prop}
\begin{proof}
From the Green's formula $Q(\mathbf{z},\mathbf{z})=0$, that together
with (\ref{eq:Qprop}) gives

\[
0=Q(\mathbf{V}+\mathbf{S}\mathbf{W},\mathbf{V}+\mathbf{S}\mathbf{W})=Q(\mathbf{V},\mathbf{V})+Q(\mathbf{S}\mathbf{W},\mathbf{S}\mathbf{W})=-i\mathbb{I}+i\mathbf{S}\mathbf{S}^{*},
\]
what furnishes the result.
\end{proof}
Consider the non-homogeneous problem (\ref{April9b}) with $f\in\mathcal{{L}}_{\gamma}^{+}(\Pi)$.
This problem has a solution $w\in{\cal H}_{\gamma}^{-}$ which
admits the asymptotic representation 
\begin{equation}
w=\sum_{j=1}^{N}\sum_{\tau=\pm}C_{j\tau}^{1}W_{j}^{\tau}+\sum_{j=1}^{N}\sum_{\tau=\pm}C_{j\tau}^{2}V_{j}^{\tau}+R,\;\; R\in{\cal H}_{\gamma}^{+}\label{tod}
\end{equation}
which is a rearrangement of the representation (\ref{April9c-1}).
This motivates the following definition of the space ${\cal H}_{\gamma}^{out}$
consisting of vector functions $w\in{\cal H}_{\gamma}^{-}$ which
admits the asymptotic representation (\ref{tod}) with $C_{j\tau}^{2}=0$.
The norm in this space is defined by 
\begin{equation*}
||w;{\cal H}_{\gamma}^{{\rm out}}||=\Big(||R;{\cal H}_{\gamma}^{+}||^{2}+\sum_{j=1}^{N}\sum_{\tau=\pm}\,|C_{j\tau}^{1}|^{2}\Big)^{1/2}.
\end{equation*}
Now, we note that the kernel in Theorem \ref{tTh1} can be equivalently
spanned by

\begin{equation*}
Z_{k}^{\tau}=W_{k}^{\tau}+\sum_{\theta=\pm}\sum_{j=1}^{N}{\bf \tilde{S}}_{k\tau}^{j\theta}V_{j}^{\theta}+\tilde{Z}_{\text{}k}^{\tau},\:\:\:\tilde{Z}_{k}^{\tau}\in{\cal H}_{\gamma}^{+},
\end{equation*}
where the incoming
and outgoing waves were interchanged (compare with (\ref{zk})) and
$\tilde{S}$ is a scattering matrix corresponding to that exchange.

\begin{theorem}\label{thrm_iso-1} For any $f\in\mathcal{L}_{\gamma}^{+}(\Pi)$,
problem {\rm (\ref{April9b})} has a unique solution $w\in{\cal H}_{\gamma}^{out}$
and the following estimate holds 
\begin{equation*}
||w;{\cal H}_{\gamma}^{{\rm out}}||\leq c||f;\mathcal{L}_{\gamma}^{+}(\Pi)||,
\end{equation*}
where the constant $c$ is independent of $\epsilon\in[0,\epsilon_{N}]$
and $|\delta|\leq\delta_{0}$. Moreover, 
\begin{equation}
iC_{j\tau}^{1}=\int_{\Pi}f\cdot\overline{Z_{j}^{\tau}}dxdy.\label{April9c-1}
\end{equation}
\end{theorem}
\begin{proof}
The proof is analogous to the proof of Theorem 3.5, presented in \cite{TrappedArXiv}.
\end{proof}
We represent ${\bf S}$ as 
\begin{equation*}
{\bf S}=\Bbb I+{\bf s},\;\;\;\mbox{or, equivalently,}\;\;{\bf S}_{j\tau}^{k\theta}=\delta_{k,j}\delta_{\tau,\theta}+{\bf s}_{j\tau}^{k\theta}.
\end{equation*}

\begin{theorem}\label{lemma_analytic-1} The scattering matrix ${\bf S}(\epsilon,\delta)$
depends analytically on small parameters $\epsilon\in[0,\epsilon_{N}]$
and $\delta\in[-\delta_{0},\delta_{0}]$. Moreover, 
\begin{equation}
{\bf s}_{j\tau}^{k\theta}=i\delta\int_{\Pi}\,\mathcal{P}{\bf w}_{j}^{\tau}\cdot\overline{{\bf w}_{k}^{\theta}}dxdy+O(\delta^{2}).\label{Apr8c}
\end{equation}
\end{theorem}
\begin{proof}
The proof follows the idea of the proof of Theorem 3.6 in \cite{TrappedArXiv}.
\end{proof}

\section{Trapped modes}

\subsection{Necessary and sufficient conditions for the existence of trapped
mode solutions\label{sub:Necessary-and-sufficient}}

Let us first introduce a value $d$ which is crucial for the formulation
of the necessary and sufficient condition for the existence of trapped
modes

\begin{equation}
d(\epsilon)=\frac{\lambda_{\epsilon}+1}{\lambda_{\epsilon}-1}=-1-i2\sqrt{\epsilon}\sqrt{2\omega_{N}}+O(\epsilon),\;\;\;|d|=1.\label{eq:d}
\end{equation}

\begin{theorem}\label{ThTrap} Problem {\rm(\ref{Dirac})} with boundary
conditions {\rm(\ref{eq:BC1INIT})}, {\rm(\ref{eq:BC2INIT})} has a non-trivial
solution in $X_{0}$ (a trapped mode), if and only if the following
matrix is degenerate 
\[
{\bf S}_{\dagger\dagger}+d(\varepsilon)\Upsilon,\;\;\;\Upsilon=\left(\begin{array}{cc}
0 & 1\\
1 & 0
\end{array}\right).
\]
\end{theorem}
\begin{proof}
A trapped mode $w\in X_{0}$, is a solution to (\ref{Dirac}) with
boundary conditions (\ref{eq:BC1INIT}), (\ref{eq:BC2INIT}), so certainly
$w\in\ker A_{-\gamma}$ and hence 
\[
w=a({\bf V}+{\bf S}{\bf W}+{\bf r})^{\top},
\]
where $a=(a_{\bullet},a_{\dagger})\in\Bbb C^{2N}$ and ${\bf V}$,
${\bf W}$ and ${\bf r}$ are the vector functions from the representation
of the kernel of $A_{-\gamma}$ in (\ref{z_sol}). Using the splitting
of vectors and the scattering matrix in $\bullet$ and $\dagger$
components, we write the above relation as 
\begin{equation*}
w=a_{\bullet}({\bf V_{\bullet}}+{\bf S}_{\bullet\bullet}{\bf W}_{\bullet}+{\bf S}_{\bullet\dagger}{\bf W}_{\dagger}+{\bf r}_{\bullet})^{T}+a_{\dagger}({\bf V_{\dagger}}+{\bf S}_{\dagger\dagger}{\bf W}_{\dagger}+{\bf S}_{\dagger\bullet}{\bf W}_{\bullet}+{\bf r}_{\dagger})^{\top}.
\end{equation*}
The first term in the right-hand side contains waves oscillating at
$\pm\infty$ and to guarantee the vanishing of this term we must
require $a_{\bullet}=0$. Since ${\bf r}$ vanishes at $\pm\infty$,
a trapped mode $w\in X_{0}$ has a representation with 
\begin{equation}
a_{\dagger}({\bf V_{\dagger}}+{\bf S}_{\dagger\dagger}{\bf W}_{\dagger}+{\bf S}_{\dagger\bullet}{\bf W}_{\bullet})^{\top}\;\;\mbox{is vanishing at \ensuremath{\pm\infty}.}\label{Apr10b}
\end{equation}
From the representations (\ref{eq:WNPWNM}) and (\ref{Apr10b}), matching
the coefficients for the increasing exponents at $\pm\infty$ we arrive
at

\begin{equation}
\frac{a_{1}}{2\mathcal{N}}\frac{\lambda_{\epsilon}+1}{\lambda_{\epsilon}}+(a_{1},a_{2}){\bf S}_{\dagger\dagger}\Big(0,\frac{1}{2\mathcal{N}}\frac{\lambda_{\epsilon}-1}{\lambda_{\epsilon}}\Big)^{\top}=0,\label{eq:d1}
\end{equation}

\begin{equation}
\frac{a_{2}}{2\mathcal{N}}\frac{\lambda_{\epsilon}+1}{\lambda_{\epsilon}}+(a_{1},a_{2}){\bf S}_{\dagger\dagger}\Big(\frac{1}{2\mathcal{N}}\frac{\lambda_{\epsilon}-1}{\lambda_{\epsilon}},0\Big)^{\top}=0,\label{eq:d2}
\end{equation}
where $a_{\dagger}=(a_{1},a_{2})$. Using the definition of $d(\epsilon)$
in (\ref{eq:d}), we write equations (\ref{eq:d1}) and (\ref{eq:d2})
as 

\begin{equation}
a_{\dagger}\Big({\bf S}_{\dagger\dagger}+d(\epsilon)\Upsilon\Big)=0.\label{eq:condd}
\end{equation}
Now, using the condition (\ref{eq:condd}), the unitarity property of matrix ${\bf S}$
and $a_{\bullet}=0$, we get 
\[
|a_{\dagger}|^{2}=|a|^{2}=|{\bf S}a|^{2}=|a_{\dagger}{\bf S}_{\dagger\bullet}|^{2}+|a_{\dagger}{\bf S}_{\dagger\dagger}|^{2}=|a_{\dagger}{\bf S}_{\dagger\bullet}|^{2}+|da_{\dagger}|^{2}.
\]
Since $|d|=1$, we have $a_{\dagger}{\bf S}_{\dagger\bullet}=0$ and
so $w$ with asymptotic representation (\ref{Apr10b}) does not contain
any oscillatory waves and is vanishing at $\pm\infty$. 
\end{proof}

\subsection{Proof of Theorem \ref{thrm1} \label{sub:Proof-of-Theorem1}}

In this section, we construct an example potential, that produces
a trapped mode. The potential is extracted from the asymptotic analysis
of a set of conditions posed on the scattering matrix ${\bf S}$.
The crucial condition concerns the augumented part of the scattering
matrix ${\bf S}$ and arrives from Theorem \ref{ThTrap}

\textcolor{black}{
\begin{equation}
\det({\bf S}_{\dagger\dagger}+d\Upsilon)=0.\label{Apr10ab}
\end{equation}
}To distinguish between small and big elements in the asymptotic analysis,
let us introduce the following notation 
\begin{align}
d(\epsilon) & =-e^{i\sigma},\;\;\;{\bf S}=\Bbb I+{\bf s}\:\:\:,\mathbf{s}=:i\delta s,\label{eq:representations}
\end{align}
where 
\begin{equation}\label{sigmaexp}
\sigma=\sqrt{\epsilon}2\sqrt{2\omega_N}+O(\epsilon^{3/2}),
\end{equation}
 from expansion (\ref{eq:d}).
 From (\ref{Apr8c}),
${\bf s}$ is of order $\delta$ and hence $s$ in (\ref{eq:representations})
is of order $1$. Now, we have three small parameters $\epsilon$,
$\sigma$ and $\delta$. 

To resolve the condition (\ref{Apr10ab}), let us list the properties
of the scattering matrix ${\bf S}$. 

First of all, the symmetry $T_{1}$ is valid for
any choice of potential $\mathcal{P}$, provided it is a real-valued
function. This symmetry imposes the following relations on the elements
of the matrix $s$: 
\begin{equation}
s_{N\theta}^{j\tau}=s_{j(-\tau)}^{N(-\theta)},\;\;\; j=1,\ldots,N-1,\;\tau,\theta=\pm,\label{eq:T1_1}
\end{equation}
\begin{equation}
s_{N+}^{N+}=s_{N-}^{N-}.\label{eq:T1_2}
\end{equation}
Those relations are illustrated in the sketch of the scattering matrix
${\bf S}$ in Figure \textcolor{black}{\ref{fig:matrixS}.}

\textcolor{black}{Secondly, requiring the symmetry of the potential
$\mathcal{P}(x,y)=\mathcal{P}(x,L-y)$, all the elements of the matrix
}${\bf S}$\textcolor{black}{{} that are odd functions of $y$ with
respect to $\frac{L}{2}$ vanish. We identify those elements
using the symmetry $T_{3}$ (Figure \ref{fig:matrixS}).}

Taking into account the unitarity property \textcolor{black}{of the matrix
}${\bf S}$ and all the above listed properties, \textcolor{black}{we
seek for $\mathcal{P}$ and small $\delta>0$ that satisfy
the relations
\begin{equation}
s_{N+}^{j\tau}=0,\;\;\; s_{N-}^{N+}=0\label{eq:Apr18a}
\end{equation}
with $\tau=\pm$, $j\in\mathcal{I}nd_{S}$, where 
\[
\mathcal{I}nd_{S}=\begin{cases}
j=1,3,\ldots,N-1 & \mbox{when }N\;\mbox{is even},\\
j=2,4,\ldots,N-2 & \mbox{when }N\;\mbox{is odd}.
\end{cases}
\]
}This choice together with the unitarity property and symmetries (\ref{eq:T1_1}),
(\ref{eq:T1_2}) yield

\[
s_{N-}^{j\tau}=0,\;\;\; s_{j\tau}^{N+}=0,\;\;\;,s_{j\tau}^{N-}=0,\;\;\; s_{N+}^{N-}=0,
\]
with $\tau=\pm$, $j=1,\,2,\ldots,N-1$
and

\textcolor{black}{
\[
|1+i\delta s_{N+}^{N+}|=1.
\]
Thus condition (\ref{Apr10ab}) becomes 
\begin{equation*}
1+i\delta s_{N+}^{N+}=\pm d,
\end{equation*}
that is equivalent to 
\begin{equation}
\Im(1+i\delta s_{N+}^{N+})=\pm\Im d.\label{Apr18b}
\end{equation}
To solve this equation, we fix the last small parameter 
\begin{equation*}
\delta=\sin\sigma,
\end{equation*}
that according to the expansion (\ref{sigmaexp}) gives
$\delta=\sqrt{\epsilon}C_{d}(1+O(\epsilon))$ with $C_d=2\sqrt{2\omega_N}$
and (\ref{Apr18b}) becomes
\begin{equation*}
\Re(s_{N+}^{N+})=\pm1,
\end{equation*}
}In particular, we seek for the potential $\mathcal{P}$ so that 
\begin{equation}
\Re(s_{N+}^{N+})=1,\label{eq:R1}
\end{equation}
\textcolor{black}{together with (\ref{eq:Apr18a}).}
\textcolor{black}{Let us use the asymptotic formula
\begin{equation}
s_{j\tau}^{k\theta}(\delta\mathcal{P})=\int_{\Pi}\,\mathcal{P}{\bf w}_{j}^{\tau}\cdot\overline{{\bf w}_{k}^{\theta}}dxdy+O(\delta)\label{Apr8cy}
\end{equation}
with $j,k=1,2,\ldots,N$ and $\tau,\theta=\pm$ which follows from
(\ref{Apr8c}) and (\ref{eq:representations}). }

\textcolor{black}{Getting together the set of conditions (\ref{eq:Apr18a})
and (\ref{eq:R1}), we obtain the following system of $2(N+1)+3$
equations when $N$ is even and of $2(N-1)+3$ equations, when $N$
is odd 
\begin{equation}
\Re s_{N+}^{j\tau}(\delta\mathcal{P})=0,\;\;\Im s_{N+}^{j\tau}(\delta\mathcal{P})=0,\;\; j\in\mathcal{I}nd_{S},\;\;\tau=\pm,\label{Apr11aa-1}
\end{equation}
\begin{equation}
\Re s_{N-}^{N+}(\delta\mathcal{P})=0,\;\;\Im s_{N-}^{N+}(\delta\mathcal{P})=0,\label{eq:Apr11bb-1}
\end{equation}
and
\begin{equation}
\Re\Big(s_{N+}^{N+}(\delta\mathcal{P})\Big)=1.\label{Apr11ab-1}
\end{equation}
}To unite the notation, we introduce \textcolor{black}{the set of indece}s:\textcolor{black}{
\begin{align}
\mathcal{I}nd & =\Big\{\alpha=(j,\tau,\Xi):\ \ \  j\in\mathcal{I}nd_{S};\,\tau=\{+,-\};\,\Xi=\{\Re,\,\Im\};\nonumber \\
 & j=N;\,\tau=-;\,\Xi=\{\Re,\,\Im\};\ \  \  j=N;\,\tau=+;\,\Xi=\Re\Big\}.\label{eq:Ind}
\end{align}
The indices with $j\in\mathcal{I}nd_{S}$ are related to equation
(\ref{Apr11aa-1}), the indices with $j=N,\,\tau=+$ correspond to
(\ref{eq:Apr11bb-1}) and the last index $(N,+,\Re)$ corresponds
to (\ref{Apr11ab-1}).}
We are looking for the potential having the following from:\textcolor{black}{
\begin{align*}
\mathcal{P}(x,y) & =\Phi(x,y)+\sum_{\alpha\in\mathcal{I}nd}\eta^{\alpha}\Psi^{\alpha}(x,y)
\end{align*}
where the functions $\Phi$, $\{\Psi^{\alpha}\}_{\alpha\in\mathcal{I}nd}$}
\textcolor{black}{are continuous, real valued with compact support
in $[-R_{0},R_{0}]\times[0,1]$. The functions are assumed to be fixed
and are subject to a set of conditions that are presented later on
in this section. The unknown coefficients $\{\eta^{\alpha}\}_{\alpha\in\mathcal{I}nd}$
can be chosen from the Banach Fixed Point Theorem. Using indices $\mathcal{I}nd$
(\ref{eq:Ind}) and the asymptotic form of the scattering matrix (\ref{Apr8cy})
we define 
\begin{equation*}
s_{\alpha}:=\Xi s_{j\tau}^{N+},\:\:\:\alpha\neq(N,+,\Re);\:\:\: s_{\alpha}:=\Re(s_{N+}^{N+}),\:\:\:\alpha=(N,+,\Re)
\end{equation*}
 and 
\begin{equation*}
\mathbf{\upsilon_{\alpha}}:=\Xi\Big({\bf w}_{j}^{\tau}\cdot\overline{{\bf w}_{N}^{+}}\Big),\:\:\:\alpha\neq(N,+,\Re);\:\:\:\mathbf{\upsilon_{\alpha}}:={\bf w}_{N}^{+}\cdot\overline{{\bf w}_{N}^{+}},\:\:\:\alpha=(N,+,\Re).
\end{equation*}
This enables us to write equations (\ref{Apr11aa-1}), (\ref{eq:Apr11bb-1})
and (\ref{Apr11ab-1}), first in the scalar form 
\begin{align*}
s_{\alpha}\Big(\delta(\Phi+\sum_{\beta\in\mathcal{I}nd}\eta^{\beta}\Psi^{\beta})\Big) & =\int_{\Pi}\,(\Phi+\sum_{\beta\in\mathcal{I}nd}\eta^{\beta}\Psi^{\beta})\upsilon_{\alpha}dxdy-\delta\mu_{\alpha}(\delta,\boldsymbol{\eta})=0,
 & \alpha\in\mathcal{I}nd,\nonumber 
\end{align*}
and then combine to a matrix form 
\begin{equation}
\mathcal{M}(\delta,\boldsymbol{\eta}):={\bf \Phi}+{\cal A}\boldsymbol{\eta}-\delta\boldsymbol{\mu}(\delta,\boldsymbol{\eta})=\delta_{(N,+,\Re)}^{\alpha},\label{matrixeq}
\end{equation}
}with a vector function $\mathcal{M}=\{\mathcal{M}_{\alpha}\}_{\alpha\in\mathcal{I}nd}$,
a vector\textcolor{black}{{} $\boldsymbol{\Phi}=\{\Phi_{\alpha}\}_{\alpha\in\mathcal{I}nd}$
with the elements 
\begin{align}
\Phi_{\alpha} & =\int_{\Pi}\,\Phi\upsilon_{\alpha}dxdy,\label{eq:phimultiplicands}
\end{align}
the matrix ${\cal A=}\{\mathcal{A}_{\alpha}^{\beta}\}_{\alpha,\beta\in\mathcal{I}nd}$
given by 
\begin{align*}
\mathcal{A}_{\alpha}^{\beta} & =\int_{\Pi}\,\Psi^{\beta}\upsilon_{\alpha}dxdy,
\end{align*}
a vector $\boldsymbol{\eta}=\{\eta^{\alpha}\}_{\alpha\in\mathcal{I}nd}$
with real unknown coefficients and a vector function $\boldsymbol{\mu}=\{\mathcal{\mu}_{\alpha}\}_{\alpha\in\mathcal{I}nd}$
that depends on $\delta$ and $\boldsymbol{\eta}$ analytically (analyticity
follows form Theorem \ref{lemma_analytic-1}). }

\textcolor{black}{Our goal is to solve system (\ref{matrixeq})
with respect to $\boldsymbol{\eta}$. We reach it in three steps.
First, we eliminate the constant $"1"$ on the right-hand side of (\ref{matrixeq}),
as $\delta_{(N,+,\Re)}^{\alpha}=1$ for $\alpha=(N,+,\Re)$, by an
appropriate choice of function $\Phi$. Secondly, we choose the functions
$\{\Psi^{\alpha}\}_{\alpha\in\mathcal{I}nd}$ in such a way that ${\cal A}$
is a unit and our system becomes nothing more than $\boldsymbol{\eta}=f(\boldsymbol{\eta})$$ $
(with a certain small function $f$) and is solvable due to the Banach Fixed
Point Theorem. }

\textcolor{black}{The choice of function $\Phi$ is the following
\begin{equation}
\Phi_{\alpha}=0,\,\,\,\alpha\neq(N,+,\Re);\:\:\:\Phi_{\alpha}=1,\,\,\,\alpha=(N,+,\Re)\label{eq:phiequations}
\end{equation}
and it is possible due to the following lemma.}

\begin{lemma} \textcolor{black}{When $2L$ is not a natural
number, then functions
\begin{equation}
\Re{\bf w}_{N}^{+}\cdot\overline{{\bf w}_{j}^{\tau}},\;\Im{\bf w}_{N}^{+}\cdot\overline{{\bf w}_{j}^{\tau}},\;\Re{\bf w}_{N}^{-}\cdot\overline{{\bf w}_{N}^{+}},\;\Im{\bf w}_{N}^{-}\cdot\overline{{\bf w}_{N}^{+}},\;{\bf w}_{N}^{+}\cdot\overline{{\bf w}_{N}^{+}}\label{eq:flinindep}
\end{equation}
with $j\in\mathcal{I}nd_{S}$, $\tau=\pm$ are linearly independent.}\textcolor{red}{{}
\label{Lemma2}}

\end{lemma}

\begin{proof} \textcolor{black}{We first note that functions (\ref{eq:flinindep})
continuously depend on $\epsilon$, so for the proof of linear independence,
it is enough to consider the limit case $\epsilon=0$. Using
the expressions (\ref{eq:soluv}) for the oscillatory waves $w_{j}^{\tau}$
and the formulas (\ref{eq:WNPWNM}) for the exponential waves ${\bf w}_{N}^{+}$
and ${\bf w}_{N}^{-}$ with their asymptotic behaviour (\ref{eq:weps+}),
(\ref{eq:weps-}), we can write 
\[
{\bf w}_{N}^{+}\cdot\overline{{\bf w}_{j}^{\tau}}=e^{-\tau i\lambda_{j}x}a_{j}(y)(C_{1}^{j\tau}+xC_{3}^{j\tau}+i(C_{2}^{j\tau}+xC_{4}^{j})),
\]
with $\tau=\pm$ and $j\in\mathcal{I}nd_{S}$ or separating real and
imaginary parts}

\textcolor{black}{
\begin{equation}
\Re{\bf w}_{N}^{+}\cdot\overline{{\bf w}_{j}^{\tau}}=a_{j}(y)\Big((C_{1}^{j\tau}+xC_{3}^{j\tau})\cos(\lambda_{j}x)+\tau(C_{2}^{j\tau}+xC_{4}^{j})\sin(\lambda_{j}x)),\label{eq:l1}
\end{equation}
\begin{equation}
\Im{\bf w}_{N}^{+}\cdot\overline{{\bf w}_{j}^{\tau}}=a_{j}(y)\Big((C_{2}^{j\tau}+xC_{4}^{j})\cos(\lambda_{j}x)-\tau(C_{1}^{j\tau}+xC_{3}^{j\tau})\sin(\lambda_{j}x)),\label{eq:l2}
\end{equation}
and for the expoential waves 
\begin{equation}
\Re{\bf w}_{N}^{+}\cdot\overline{{\bf w}_{N}^{-}}=\frac{\omega_{N}}{L}(-4x^{2}+\frac{4\mbox{sign}(\kappa_{N})}{\omega_{N}}x-\frac{2}{\omega_{N}^{2}}+4),\label{eq:l3}
\end{equation}
\begin{equation}
\Im{\bf w}_{N}^{+}\cdot\overline{{\bf w}_{N}^{-}}=\frac{\omega_{N}}{L}(8x-\frac{4\mbox{sign}(\kappa_{N})}{\omega_{N}}),\label{eq:l4}
\end{equation}
\begin{equation}
{\bf w}_{N}^{+}\cdot\overline{{\bf w}_{N}^{+}}=\frac{\omega_{N}}{L}(2x^{2}-\frac{2\mbox{sign}(\kappa_{N})}{\omega_{N}}x+\frac{1}{\omega_{N}^{2}}+2),\label{eq:l5}
\end{equation}
\[
\]
with 
\begin{equation}
a_{j}(y)=\frac{\omega_{N}}{4L\lambda_{j}}\cos\Big((\kappa_{N}-\kappa_{j})y\big)=\frac{\omega_{N}}{4L\lambda_{j}}\cos(\frac{\pi}{L}(N-j)y)\label{eq:aj}
\end{equation}
being an even function with respect to $\frac{L}{2}$, and constants}

\textcolor{black}{
\[
C_{1}^{j\tau}=1+\mbox{sign}(\kappa_{j})\frac{\kappa_{j}}{\omega}+\frac{\tau\lambda_{j}}{\omega^{2}},\;\;\; C_{2}^{j\tau}=\mbox{sign}(\kappa_{j})\frac{\tau\lambda_{j}}{\omega}-\frac{\kappa_{j}}{\omega},
\]
}
\[
C_{3}^{j\tau}=-\mbox{sign}(\kappa_{j})\frac{\tau\lambda_{j}}{\omega},\;\;\; C_{4}^{j}=1+\mbox{sign}(\kappa_{j})\frac{\kappa_{j}}{\omega}.
\]
\textcolor{black}{From the form (\ref{eq:aj}), functions $\{a_{j}(y)\}_{j\in\mathcal{I}nd}$
are linearly independent. Still for a fixed $j$ , functions $\Re{\bf w}_{N}^{+}\cdot\overline{{\bf w}_{j}^{+}}$,
$\Re{\bf w}_{N}^{+}\cdot\overline{{\bf w}_{j}^{-}}$, $\Im{\bf w}_{N}^{+}\cdot\overline{{\bf w}_{j}^{+}}$,
$\Im{\bf w}_{N}^{+}\cdot\overline{{\bf w}_{j}^{-}}$ in (\ref{eq:l1})
and (\ref{eq:l2}) could be linearly dependent, as they are linear
combinations of four types of functions $\cos(\lambda_{j}x)$, $x\cos(\lambda_{j}x)$,
$\sin(\lambda_{j}x)$ and $x\sin(\lambda_{j}x)$. However this possiblity
is ruled out as the determinant with the coefficients of the composite
functions $\cos(\lambda_{j}x)$, $x\cos(\lambda_{j}x)$, $\sin(\lambda_{j}x)$
and $x\sin(\lambda_{j}x)$ is non-zero. It follows that (\ref{eq:l1})
and (\ref{eq:l2}) are linearly independent. Finally the functions
(\ref{eq:l3}), (\ref{eq:l4}), (\ref{eq:l5}) belong to $\{1,\, x,\, x^{2}\}$
and are linear independent as $\Re{\bf w}_{N}^{+}\cdot\overline{{\bf w}_{N}^{-}}+2{\bf w}_{N}^{+}\cdot\overline{{\bf w}_{N}^{+}}=8$.}

\end{proof} 

\textcolor{black}{By Lemma \ref{Lemma2}, all the functions $v_{\alpha}$ in (\ref{eq:phimultiplicands}) are linearly independent.
It follows that it is possible to choose $\Phi$ so that (\ref{eq:phiequations})
holds and equations (\ref{matrixeq}) is
\begin{equation}
{\cal A}\boldsymbol{\eta}-\delta\boldsymbol{\mu}(\delta,\boldsymbol{\eta})=0.\label{eq:matrixeqlin}
\end{equation}
Now we set matrix ${\cal A}$ to be unit, that is its elements fulfill
the conditions
\begin{equation}
\mathcal{A}_{\alpha}^{\beta}=\delta_{\alpha,\beta},\,\,\,\alpha,\beta\in\mathcal{I}nd.\label{eq:Aunit}
\end{equation}
Again using Lemma \ref{Lemma2}, it is possible to choose functions
$\{\Psi^{\alpha}\}_{\alpha\in\mathcal{I}nd}$ so that the conditions
(\ref{eq:Aunit}) are fulfilled and (\ref{eq:matrixeqlin}) reads
\begin{equation}
\boldsymbol{\eta}=\delta\boldsymbol{\mu}(\delta,\boldsymbol{\eta}).\label{eq:contraction}
\end{equation}
Now, as $\delta$ is small, the operator on the right hand side of
equation (\ref{eq:contraction}) is a contraction operator, moreover
$\boldsymbol{\mu}$ is analytic in $\delta$ and $\boldsymbol{\eta}$
so from Banach Fixed Point Theorem equation (\ref{eq:contraction})
is solvable for $\boldsymbol{\eta}$.}

We have just shown that it is possible to choose functions \textcolor{black}{$\Phi$,
$\{\Psi^{\alpha}\}_{\alpha\in\mathcal{I}nd}$} and tune parameters
\textcolor{black}{$\boldsymbol{\eta}$ in such a way that the potential
$\mathcal{P}$ produces a trapped mode.}

A numerical example of a potential (leading term $\Phi$) that produces
a trapped mode is 

\[
\mathcal{P}(x,y)\approx\Phi(x,y)=e^{-(\frac{y-0.67}{0.2})^{2}}(0.54e^{-(x+0.14)^{2}}-e^{-(x+0.32)^{2}}+0.54e^{-(x+0.49)^{2}})
\]
for a nanoribbon of width $L=1.33$ ($\frac{L}{2}=0.67$) and energy
$\omega\approx\omega_{2}=1.57$.

\subsection{Proof of Theorem\label{sub:Proof-of-Theorem2} }

In this section, we present the proof to the theorem about the multiplicity
of the Trapped modes. We show that (i) there are no trapped modes for
energies that are slightly bigger than the thresholds or are far from
the thresholds, (ii) the multiplicity of a trapped mode, for energies
slightly less than a threshold, is no more than one.

\begin{proof} 

(i) Assume the contrary: there exist a trapped mode solution, which
belongs to $X_{0}$. According to the Theorem \ref{T3s2} and Proposition
\ref{Prop1} (ii), $\lambda$ in the strip $\{|\Im\lambda|\leq\gamma\}$,
in the exponential factor of the solutions (\ref{uosc}), (\ref{eq:solUV-2-1}),
\[
w(x,y)=e^{i\lambda x}(\mathcal{U}(y),\mathcal{V}(y),\mathcal{U}'(y),\mathcal{V}'(y))
\]
are real, it follows that $w\in\mathcal{H}_{\gamma}^{+}$. However
according to the Theorem \ref{T1a}, the operator $\mathcal{D}+(\delta\mathcal{P}-\omega_{\epsilon})\mathcal{I}\,:\, X_{\gamma}^{\pm}\rightarrow L_{\gamma}^{\pm}(\Pi)$
is an isomorphism, it follows that the only solution to $\Big(\mathcal{D}+(\delta\mathcal{P}-\omega_{\epsilon})\mathcal{I}\Big)w=0$
is $w=0$.

(ii) There exists at least one trapped mode given in Sect.
\ref{thrm1}. Assume now, that we have two trapped modes
$w_{1}\neq w_{2}$. From Theorem \ref{T3s2} and
Proposition \ref{Prop1} (i), which states that there are exactly
two solutions $w(x,y)=e^{i\lambda x}(\mathcal{U}(y),\mathcal{V}(y),\mathcal{U}'(y),\mathcal{V}'(y))$
with complex $\lambda$ in the strip $\{|\Im\lambda|\leq\gamma\}$,
it follows that the trapped mode is of the form

\begin{align*}
w_{j} & =C_{j}e^{i\lambda_{\epsilon}x}(\text{\ensuremath{\mathcal{U}}}_{N}^{+}(y),\mathcal{V}_{N}^{+}(y),\mathcal{U}_{N}^{'+}(y),\mathcal{V}_{N}^{'+}(y))\\
 & +D_{j}e^{-\lambda_{\epsilon}x}(\mathcal{U}_{N}^{-}(y),\mathcal{V}_{N}^{-}(y),\mathcal{U}_{N}^{'-}(y),\mathcal{V}_{N}^{'-}(y))+R_{j},
\end{align*}
with $j=1,2$ and $R_{j}\in\mathcal{H}_{\gamma}^{+}$. Consider the
following linear combination of trapped modes $\omega_{1}$ and $\omega_{2}$
\begin{align*}
w_{3} & =w_{1}-\frac{C_{1}}{C_{2}}w_{2}=(D_{1}-\frac{C_{1}}{C_{2}}D_{2})e^{-i\lambda_{\epsilon}^{-}x}(\mathcal{U}_{N}^{-}(y),\text{\ensuremath{\mathcal{V}}}_{N}^{-}(y),\mathcal{U}_{N}^{'-}(y),\mathcal{V}_{N}^{'-}(y))\\
 & +(R_{1}-\frac{C_{1}}{C_{2}}R_{2})\in X_{\gamma}^{+},
\end{align*}
however form Theorem \ref{T1a} as the operator $\mathcal{D}+(\delta\mathcal{P}-\omega_{\epsilon})\mathcal{I}\,:\, X_{\gamma}^{+}\rightarrow L_{\gamma}^{+}(\Pi)$
is an isomorphism, it follows that $w_{1}=\frac{C_{1}}{C_{2}}w_{2}$.

\end{proof} 

\section*{Acknowledgement}
V. Kozlov and A. Orlof acknowledge support of the Link\"{o}ping Univeristy. The authors thank I. V. Zozoulenko for the discussion on physical aspects of the paper. S. A. Nazarov acknowledges financial support from The Russian Science Foundation (Grant 14-29-00199). 

\appendix

\section{Norm\label{sec:Ellipticity}}
\begin{proof}
Let us consider only a part of $\int_{\Pi}|\mathcal{D}(u,v,u',v')^{\top}|^{2}dxdy$,
that contains elements with functions $u$ and $u'$, namely

\[
\int_{\Pi}|(i\partial_{x}-\partial_{y})u|^{2}+|(-i\partial_{x}-\partial_{y})u'|^{2}dxdy
\]
\[
=\int_{\Pi}\Big(|\nabla u|^{2}+|\nabla u'|^{2}\Big)dxdy+i\int_{\Pi}\Big(\overline{u}_{x}u_{y}-\overline{u}_{y}u_{x}-\overline{u}_{x}'u_{y}'+\overline{u}_{y}'u_{x}'\Big)dxdy.
\]
We show that due to the boundary conditions (\ref{eq:BC1INIT}) and
(\ref{eq:BC2INIT})
\[
\int_{\Pi}\Big(\overline{u}_{x}u_{y}-\overline{u}_{y}u_{x}-\overline{u}_{x}'u_{y}'+\overline{u}_{y}'u_{x}'\Big)dxdy=0.
\]
Using integration by parts, we get
\begin{align*}
\int_{\Pi}\Big(\overline{u}_{x}u_{y}-\overline{u}_{y}u_{x}-\overline{u}_{x}'u_{y}'+\overline{u}_{y}'u_{x}'\Big)dxdy & =\int_{-\infty}^{+\infty}\overline{u}_{x}u-\overline{u}u_{x}-\overline{u}_{x}'u'+\overline{u}'u_{x}'\Big|_{0}^{L}dx\\
=\int_{-\infty}^{+\infty}\bigg(\overline{u}_{x}(x,L)\Big(u(x,L)-ie^{i2\pi L}u'(x,L)\Big) & -u_{x}(x,L)\Big(\overline{u}(x,L)+ie^{-i2\pi L}\overline{u}'(x,L)\Big)\\
-\overline{u}_{x}(x,0)\Big(u(x,0)-iu'(x,0)\Big) & +u_{x}(x,0)\Big(\overline{u}(x,0)+i\overline{u}'(x,0)\Big)\bigg)dxdy=0,
\end{align*}
where the last two equalities are consequence of the boundary conditions
(\ref{eq:BC1INIT}) and (\ref{eq:BC2INIT}). In a similar way, one
can show that the remaining part of $\int_{\Pi}|\mathcal{D}(u,v,u',v')^{\top}|^{2}dxdy$,
that contains elements with functions $v$ and $v'$ is 
\[
\int_{\Pi}\Big(|(i\partial_{x}+\partial_{y})v|^{2}+|(-i\partial_{x}+\partial_{y})v'|^{2}\Big)dxdy=\int_{\Pi}\Big(|\nabla v|^{2}+|\nabla v'|^{2}\Big)dxdy,
\]
consequently 
\[
\int_{\Pi}|\mathcal{D}(u,v,u',v')^{\top}|^{2}dxdy=\int_{\Pi}\Big(|\nabla u|^2+|\nabla u'|^{2}+|\nabla v|^2+|\nabla v'|^{2}\Big)dxdy.
\]

\end{proof}

\section{The Mandelstam radiation condition\label{sec:Mandelstam-radiation-condition}}

Here we want to clarify the splitting of waves in two classes (outgoing/incoming)
according to the appearance of the $\pm i$ in (\ref{eq:bioosc}).
To do this we use the Mandelstam radiation conditions whi\textcolor{black}{ch
}defines outgoing and incoming waves by the direction of the energy
transfer \cite{Mandelstam,na569,VoBa}.

Let us write the original system (\ref{eq:DiractotINIT}) in the
form
\begin{equation}
\mathcal{D}\mathbf{w}=i\partial_{t}\mathbf{w},\:\:\:\mathbf{w}=e^{-i\omega t}w,\:\:\:\mathbf{w}=(\mathbf{u},\mathbf{v},\mathbf{u'},\mathbf{v'}).\label{eq:Mand}
\end{equation}
The energy flux through the boundary is defined as
\[
-\frac{d}{dt}\int_{\Omega}|\partial_{t}\mathbf{w}|^{2}dxdy.
\]
Using relations (\ref{eq:Mand}) and performing partial integration
we get 

\begin{align*}
-\frac{d}{dt}\int_{\Omega}|\partial_{t}\mathbf{w}|^{2}dxdy & =-|\omega|^{2}\int_{\Omega}\Big(\mathbf{\overline{\mathbf{w}}}\partial_{t}\mathbf{w}+\mathbf{w}\partial_{t}\overline{\mathbf{w}}\Big)dxdy\\
=-|\omega|^{2}\int_{\Gamma}\Big(\overline{\mathbf{u}}\mathbf{v}+\mathbf{\overline{\mathbf{v}}\mathbf{u}}-(\overline{\mathbf{u}}'\mathbf{v'}+\mathbf{\overline{\mathbf{v}}'\mathbf{u}'}), & -i(\overline{\mathbf{u}}\mathbf{v}-\mathbf{\overline{\mathbf{v}}\mathbf{u}}+\overline{\mathbf{u}}'\mathbf{v'}-\mathbf{\overline{\mathbf{v}}'\mathbf{u}'})\Big)\cdot(n_{x},n_{y})ds,
\end{align*}
where $\Gamma$ is the boundary of area $\Omega$. Consider the energy
flux through the cross-section, that is choose $(n_{x},n_{y})=(1,0)$, then the last formula
is equal to 

\[
-\frac{d}{dt}\int_{\Omega}|\partial_{t}\mathbf{w}|^{2}dxdy=-|\omega|^{2}\int_{0}^{L}\Big(\overline{u}v+\overline{v}u-(\overline{u}'v'+\overline{v}'u')\Big)dy=-i|\omega|^{2}q(w,w),
\]
where the last equality comes from the the definition of q-form (\ref{eq:q}).
Accordingly the energy transfer along the nanoribbon is proportional
to $-iq$, which is $\pm1$ for $q=\pm i$. It follows that the value
of the q-form defines direction of wave propagation, namely $q=i$ describes
waves propagating from from $-\infty$ to $+\infty$ and $q=-i$ those
from $+\infty$ to $-\infty$. This leads to the definition of outgoing/incoming
waves (\ref{W1a}), (\ref{W1b}) as those travelling to $\pm\infty$
and from $\pm\infty$ .


\begin{thebibliography}{10}
\bibitem{Ando} T. Ando, T. Nakanishi, \textit{Impurity Scattering
in Carbon Nanotubes - Absence of Back Scattering}, J. Phys. Soc. Jpn.
\textbf{67}, 1704-1713 (1998).

\bibitem{Brey} L. Brey, H. A. Fertig, \textit{Electronic states of
graphene nanoribbons studied with the Dirac equation}, Phys. Rev.
\textbf{B} 73, 235411 (2006).

\bibitem{CastroNeto} A. H. Castro Neto, F. Guinea, N. M. R. Peres,
K. S. Novoselov, A. K. Geim, \textit{The electronic properties of
graphene}, Rev. Mod. Phys. 81, 109 (2009).

\bibitem{Chen}Q. Chen, L. Ma, J. Wang, \textit{Making graphene nanoribbons:
a theoretical exploration}. WIREs Comput. Mol. Sci. 6 (3), 243–254
(2016).

\bibitem{naKamotsky} I. V. Kamotskii, S. A. Nazarov	\textit{An augmented scattering matrix and an exponentially decreasing solution of elliptic boundary-value problem in the domain with cylindrical outlets} Zap. Nauchn. Sem. POMI, 264, 66-82, (2000).

\bibitem{KozlovBV} V. Kozlov, V. Maz'ya, \textit{Differential equations
with operator coefficients with applications to boundary value problems
for partial differential equations}. Springer Monographs in Mathematics.
Springer-Verlag, Berlin, 1999.

\bibitem{TrappedArXiv} V. A. Kozlov, S. A. Nazarov, A. Orlof, \textit{Trapped
modes in zigzag graphene nanoribbons}, arXiv:1701.05795 {[}math-ph{]}
(2017).

\bibitem{Mandelstam}L. I. Mandelstam, \textit{ Lectures on Optics, Relativity,
and Quantum Mechanics}, 2, AN SSSR, Moscow (1947).

\bibitem{Macucci} P. Marconcini, M. Macucci, \textit{The k · p method
and its application to graphene, carbon nanotubes and graphene nanoribbons:
the Dirac equation, }La Rivista del Nuovo Cimento 45 (8-9), 489-584
(2011).

\bibitem{Nazarov1} S. A. Nazarov, \textit{Asymptotic expansions of
eigenvalues in the continuous spectrum of a regularly perturbated
quantum waveguide}. Theoretical and Mathematical Physics 167 (2),
239-262 (2011) (English transl.: Theoretical and Hathematical Physics.
167 (2), 606-627 (2011).

\bibitem{Nazarov2} S. A. Nazarov, \textit{Enforced stability of an
eigenvalue in the continuous spectrum of a waveguide with an obstacle.}
Zh. Vychisl. Mat. i Mat. Fiz. 52 (3), 521-538 (2012) (English transl.:
Comput. Math. and Math. Physics. 52 (3), 448-464 (2012).

\bibitem{Nazarov3} S. A. Nazarov, \textit{Enforced stability of a
simple eigenvalue in the continuous spectrum of a waveguide.} Funct.
Anal. Appl. 47 (3), 195-209 (2013).

\bibitem{na569} S. A. Nazarov, \textit{Umov-Mandelstam radiation conditions in elastic periodic waveguides} SB MATH, 205 (7), 953–982 (2014).

\bibitem{NaPl}S. A. Nazarov, B. A. Plamenevsky, \textit{Elliptic problems
in domains with piecewise smooth boundaries}. Walter de Gruyter, Berlin,
New York (1994).

\bibitem{Novoselsov}K. S. Novoselov, A.K. Geim, S.V. Morozov, D.
Jiang, Y. Zhang, S. V. Dubonos, I. V. Grigorieva, A. A. Firsov, \textit{Electric
Field Effect in Atomically Thin Carbon Films}, Science \textbf{306},
666 (2004).

\bibitem{VoBa} I.I. Vorovich, V.A. Babeshko \textit{Dynamical mixed problems of elasticity theory for nonclassical domains}, Nauka, Moscow 320, (1979) (Russsian).

\end{thebibliography}
\end{document}